\documentclass[11pt]{article}

\usepackage[ruled]{algorithm}
\usepackage{algpseudocode}
\usepackage{graphicx}
\usepackage{amssymb}
\usepackage[isolatin]{inputenc}
\usepackage[OT1]{fontenc}   
\usepackage{dsfont}
\usepackage{epsfig}
\usepackage{wide}
\usepackage{url}
\usepackage{ifthen}
\usepackage{latexsym}
\usepackage{array}
\usepackage{wrapfig}
\usepackage{float}
\usepackage{setspace}
\usepackage{times}
\usepackage{float}
\floatstyle{ruled}
\usepackage{boxedminipage}
\usepackage{graphicx, subfigure}
\usepackage{epsfig}
\usepackage{ntheorem}
\usepackage{latexsym}
\usepackage{amssymb}

\alglanguage{pseudocode}

\newfloat{algorithm}{tbp}{loa}
\floatname{algorithm}{Algorithm}

\newtheorem{theorem}{Theorem}
\newtheorem{lemma}{Lemma}
\newtheorem{remark}{Remark}

\newcommand{\ie}{\emph{i.e., }}

\newcommand{\N}{\mathds{N}}

\newcommand{\BEGLIST}{\begin{list}{}{\partopsep -3pt \parsep -2pt}}
\newcommand{\ENDLIST}{\end{list}}

\newenvironment{proof}{{\bf Proof. } }{{\hfill $\Box$}\vspace{.5pc}}

\newenvironment{theo}[1]{
\begin{theorem}#1}{\end{theorem}
}

\newenvironment{lem}[1]{
\begin{lemma}#1}{\end{lemma}
}

\newenvironment{rem}[1]{
\begin{remark}#1}{\end{remark}
}

\newenvironment{keyw}{{\noindent \bf Keywords:}}{\medskip}

\date{}

\title{Optimal Grid Exploration by Asynchronous Oblivious Robots\\
  {
  {
  }
  }
  }
\author{
  St\'ephane Devismes\thanks{VERIMAG UMR 5104, Universit\'e Joseph Fourier, Grenoble (France)}
  \and
  Anissa Lamani\thanks{MIS, Universit\'e de Picardie Jules Verne (France)}
  \and
  Franck Petit\thanks{LIP6, UPMC Sorbonne Universit\'es (France)}
  \and
  Pascal Raymond$^*$ 
  \and
  S\'ebastien Tixeuil$^\ddag$ 
}

\begin{document}

\maketitle

\begin{abstract}
  We propose optimal ({\em w.r.t.} the number of robots) solutions for
  the {\em deterministic terminating exploration} ({\em exploration}
  for short) of a grid-shaped network by a team of $k$ asynchronous
  oblivious robots in the asynchronous non-atomic model, so-called
  CORDA.

  In more details, we first consider the ATOM model.  We show that it
  is impossible to explore a grid of at least three nodes with less
  than three robots. Next, we show that it is impossible to explore a
  $(2,2)$-Grid with less than 4 robots, and a $(3,3)$-Grid with less
  than 5 robots, respectively.  The two first results hold for both
  deterministic and probabilistic settings, while the latter holds
  only in the deterministic case.  ATOM being strictly stronger than
  CORDA, all these impossibility results also hold in CORDA.

  Then, we propose deterministic algorithms
  in CORDA to exhibit the optimal number of robots allowing to explore
  of a given grid.  Our results show that except in two particular
  cases, 3 robots are necessary and sufficient to deterministically
  explore a grid of at least three nodes. The optimal number of robots
  for the two remaining cases is: 4 for the $(2,2)$-Grid and 5 for the
  $(3,3)$-Grid, respectively.
\end{abstract}

\begin{keyw}
Exploration, Grid, Oblivious Robots, CORDA model.
\end{keyw}



\section{Introduction}

We consider autonomous robots~\cite{PRT11c,SY99j} that are endowed with motion actuators
and visibility sensors, but that are otherwise unable to communicate.
Those robots must collaborate to solve a collective task, here the
\emph{deterministic terminating grid exploration} ({\em exploration}
for short), despite being limited with respect to input from the
environment, asymmetry, memory, etc.

So far, two universes have been studied: the {\em continuous
  two-dimensional Euclidean space} and the {\em discrete universe}.
In the former, robot entities freely move on a plane using visual
sensors with perfect accuracy that permit to locate all other robots
with infinite precision (see \emph{e.g.},
\cite{DFSY10c,DLP08j,FPSW08j,SDY09j,SY99j}).  
In the latter, the space
is partitioned into a finite number of locations, conventionally
represented by a graph, where the nodes represent the possible
locations that a robot can take and the edges the possibility for a
robot to move from one location to another (\emph{e.g.},
\cite{BBMR08c,BMPT10c,BMPT11c,Devismes2010,DPT09c,FIPS07c,FIPS08c,IIKO10c,KLOT11c,KMP08j,LPT10c}).

In this paper, we pursue research in the discrete universe and focus
on the \emph{exploration problem} when the network is an anonymous
unoriented grid, using a team of autonomous mobile robots.
Exploration requires that robots explore the grid and stop when the
task is complete. In other words, every node of the grid must be
visited by at least one robot and the protocol eventually terminates,
{\em i.e.}, every robot eventually stays idle forever.

The robots we consider are unable to communicate, however they can
sense their environment and take decisions according to their local
view. We assume anonymous and uniform robots (\textit{i.e.}, they
execute the same protocol and there is no way to distinguish between
them using their appearance). In addition, they are oblivious,
\textit{i.e.}, they do not remember their past actions.  In this
context, robots asynchronously operate in cycles of three phases:
Look, Compute, and Move. In the first phase, robots observe their
environment in order to get the position of all other robots in the
grid. In the second phase, they perform a local computation using the
previously obtained view and decide their target destination to which
they move during the last phase.

The fact that the robots have to stop after the exploration process
implies that the robots somehow have to remember which part of the
graph has been explored. Nevertheless, under this weak scenario,
robots have no memory and thus are unable to remember the various
steps taken before. In addition, they are unable to communicate
explicitly. Therefore the positions of the other robots are the only
way to distinguish the different stages of the exploration
process. The main complexity measure is then the minimal number of
required robots.  Since numerous symmetric configurations induce a
large number of required robots, minimizing the number of robots turns
out to be a difficult problem. As a matter of fact, in~\cite{FIPS08c},
it is shown that, in general, $\Omega(n)$ robots are necessary to
explore a tree network of $n$ nodes deterministically.

\paragraph{Related Work.}
In~\cite{FIPS07c}, authors proved that no deterministic exploration is
possible on a ring when the number of robots $k$ divides the number of
nodes $n$.  In the same paper, the authors proposed a deterministic
algorithm that solves the problem using at least $17$ robots provided
that $n$ and $k$ are co-prime.
In~\cite{LPT10c}, Lamani \emph{et al.} proved that there exists no
deterministic protocol that can explore an even sized ring with $k\leq
4$ robots, even in the atomic model, so-called ATOM~\cite{SY99j}. In
this model, robots execute their Look, Compute and Move phases in an
atomic manner, \ie every robot that is activated at instant $t$
instantaneously executes a full cycle between $t$ and
$t+1$. Impossibility results in ATOM naturally extend in the
asynchronous non-atomic model, so-called CORDA~\cite{P01cb}.  Lamani
\emph{et al.}  also provide in~\cite{LPT10c} a deterministic protocol
using five robots and performing in CORDA, provided that five and $n$
are co-prime.  By contrast, four robots are necessary and sufficient
to solve the {\em probabilistic} exploration of any ring of size at
least 4 in ATOM~\cite{DPT09c,Devismes2010}.

To our knowledge, grid-shaped networks were only considered in the
context of anonymous and oblivious robot exploration~\cite{BBMR08c,BMPT11c}
for a variant of the exploration problem where robots perpetually
explore all nodes in the grid (instead of stopping after exploring the
whole network).  Also, contrary to this paper, the protocols presented
in~\cite{BBMR08c} make use of a common sense of direction for all
robots (common north, south, east, and west directions) and assume an
essentially synchronous scheduling.

\paragraph{Contribution.}
In this paper, we propose optimal ({\em w.r.t.} the number of robots)
solutions for the deterministic terminating exploration of a
grid-shaped network by a team of $k$ asynchronous oblivious robots in
the asynchronous and non-atomic CORDA model.

In more details, we first consider the ATOM model, which is a strictly
stronger model than CORDA.  We show that it is impossible to explore a
grid of at least three nodes with less than three robots. Next, we
show that it is impossible to explore a $(2,2)$-Grid with less than 4
robots, and a $(3,3)$-Grid with less than 5 robots, respectively.  The
two first results hold for both deterministic and probabilistic
settings, while the latter holds only in the deterministic case.  Note
also that these impossibility results naturally extend to CORDA.

Then, we propose several deterministic algorithms in CORDA to exhibit
the optimal number of robots allowing to explore of a given grid.  Our
results show that except in two particular cases, 3 robots are
necessary and sufficient to deterministically explore a grid of at
least three nodes. The optimal number of robots for the two remaining
cases is: 4 for the $(2,2)$-Grid and 5 for the $(3,3)$-Grid,
respectively.

The above results show that, perhaps surprisingly, exploring a grid is
easier than exploring a ring. In the ring, deterministic solutions
essentially require five robots~\cite{LPT10c} while probabilities
enable solutions with only four robots~\cite{DPT09c,Devismes2010}. In
the grid, three robots are necessary and sufficient in the general
case even for deterministic protocols, while particular instances of
the grid do require four or five robots. Also, deterministically
exploring a grid requires no primality condition while
deterministically exploring a ring expects the number $k$ of robots to
be co-prime with $n$, the number of nodes.

\paragraph{Roadmap.}

Section~\ref{sect:model} presents the system model and the problem to
be solved.  Lower bounds are shown in Section~\ref{Bounds}.  The
deterministic general solution using three robots is given
in~Section~\ref{sec:3R}, the special case with five robots is studied
in~Section~\ref{sec:5R}.  Section~\ref{conclu} gives some concluding
remarks. 

\section{Preliminaries}\label{sect:model}

\paragraph{Distributed Systems.}  
We consider systems of autonomous mobile entities called {\em agents}
or {\em robots} evolving in a {\em simple unoriented connected graph}
$G = (V,E)$, where $V$ is a finite set of nodes and $E$ a finite set
of edges.  In $G$, nodes represent locations that can be sensed by
robots and edges represent the possibility for a robot to move from
one location to another.  We assume that $G$ is an $(i,j)$-{\em Grid}
(or a Grid, for short) where $i$, $j$ are two positive integers, {\em
  i.e.}, $G$ satisfies the following two conditions:
\begin{enumerate}
\item $|V| = i \times j$, and
\item there exists an order on the nodes of $V$, $v_1,
  \ldots, v_{i\cdot j}$, such that:
  \begin{itemize}
  \item $\forall x \in [1..i\times j]$, $(x \bmod i) \neq 0
    \Rightarrow \{v_x,v_{x+1}\} \in E$, and
  \item $\forall y \in [1..i\times (j-1)]$, $\{v_y,v_{y+i}\} \in
    E$.
  \end{itemize}
\end{enumerate}

We denote by $n= i\times j$ the number of nodes in $G$. We denote by
$\delta(v)$ the degree of node $v$ in $G$. Nodes of the grid are
anonymous (we may use indices, but for notation purposes
only). Moreover, given two neighboring nodes $u$ and $v$, there is no
explicit or implicit labeling allowing the robots to determine whether
$u$ is either on the left, on the right, above, or below $v$. Remark
that an $(i,j)$-{\em Grid} and a $(j,i)$-{\em Grid} are
isomorphic. Hence, as the nodes are anonymous, we cannot distinguish
an $(i,j)$-{\em Grid} from a $(j,i)$-{\em Grid}. So, without loss of
generality, we always consider $(i,j)$-{\em Grids}, where $i\leq
j$. Note also that any $(1,j)$-{\em Grid} is isomorphic to a chain. In
any $(i,j)$-{\em Grid}, if $i = 1$, then either the grid consists of
one single node, or two nodes are of degree 1 and all other nodes are
of degree 2; otherwise, when $i>1$, four nodes are of degree 2 and all
other nodes are of degree either 3 or 4.  In any grid, the nodes of
smallest degree are called {\em corners}. In any $(1,j)$-{\em Grid}
with $j>1$, the unique chain linking the two corners is called the
{\em borderline}. In any $(i,j)$-{\em Grid} such that $i>1$, there
exist four chains $v_1$, \ldots, $v_m$ of length at least 2 such that
$\delta(v_1) = \delta(v_m) = 2$, and $\forall x, 1 < x < m$,
$\delta(v_x) = 3$, these chains are also called the {\em borderlines}.

\paragraph{Robots.} 
Operating on $G$ are $k \leq n$ robots. The robots
do not communicate in an explicit way; however they see the position
of the other robots and can acquire knowledge based on this information.
We assume that the robots cannot remember any previous observation nor
computation performed in any previous step.  Such robots are said to
be \emph{oblivious} (or \emph{memoryless}).

Each robot operates according to its (local) {\em program}.  We call
\emph{protocol} a collection of $k$ \emph{programs}, each one
operating on one single robot. Here we assume that robots are
\emph{uniform} and \emph{anonymous}, {\em i.e.}, they all have the
same program using no local parameter (such as an identity) that could
permit to differentiate them.
The program of a robot consists in executing {\em
  Look\mbox{-}Compute\mbox{-}Move cycles} infinitely many times.  That
is, the robot first observes its environment (Look phase).  Based on
its observation, a robot then decides to move or stay idle (Compute
phase).  When a robot decides to move, it moves from its current node
to a neighboring node during the Move phase.

\paragraph{Computational Model.}

We consider two models: the semi-asynchronous and atomic model,
ATOM~\cite{EP09j,SY99j} and the asynchronous non-atomic model,
CORDA~\cite{FIPS07c,P01cb}.  In both models, time is represented by an
infinite sequence of instants 0, 1, 2, \dots\ No robot has access to
this global time. Moreover, every robot executes cycles infinitely
many times. Each robot performs its own cycles in sequence. However,
the time between two cycles of the same robot and the interleavings
between cycles of different robots are decided by an {\em
  adversary}. As a matter of facts, we are interested in algorithms
that correctly operate despite the choices of the adversary. In
particular, our algorithms should also work even if the adversary
forces the execution to be fully sequential or fully synchronous.

In ATOM, each Look-Compute-Move cycle execution is assumed to be
\emph{atomic}: every robot that is activated (by the adversary) at
instant $t$ instantaneously executes a full cycle between $t$ and
$t+1$.

In CORDA, Look-Compute-Move cycles are performed asynchronously by
each robot: the time between Look, Compute, and Move operations is
finite yet unbounded, and is decided by the adversary. The only
constraint is that both Move and Look are instantaneous.

Remark that in both models, any robot performing a Look operation sees
all other robots on nodes and not on edges. However, in the CORDA, a
robot $\mathcal R$ may perform a Look operation at some time $t$,
perceiving robots at some nodes, then Compute a target neighbor at
some time $t'>t$, and Move to this neighbor at some later time
$t''>t'$ in which some robots are at different nodes from those
previously perceived by $\mathcal R$ because in the meantime they
moved. Hence, robots may move based on significantly outdated
perceptions.

Of course, ATOM is stronger than CORDA. So, to be
as general as possible, in this paper, our impossibility results are
written assuming ATOM, while our algorithms assume CORDA.

\paragraph{Multiplicity.}
We assume that during the Look phase, every robot can perceive whether
several robots are located on the same node or not.  This ability is
called \emph{Multiplicity Detection}.  We shall indicate by $d_i(t)$
the multiplicity of robots present in node $u_i$ at instant $t$.  We
consider two kinds of multiplicity detection: the {\em strong} and
{\em weak} multiplicity detections.

Under the {\em weak} multiplicity detection, for every node $u_i$,
$d_i$ is a function $\N \mapsto \{\circ,\bot,\top\}$ defined as
follows: $d_i(t)$ is equal to either $\circ$, $\bot$, or $\top$
according to $u_i$ contains none, one or several robots at time instant
$t$. If $d_i(t) = \circ$, then we say that $u_i$ is
{\em free} at instant $t$, otherwise $u_i$ is said {\em occupied} at
instant $t$. If $d_i(t) = \top$, then we say that $u_i$ contains a
{\em tower} at instant $t$.

Under the {\em strong} multiplicity detection, for every node $u_i$,
$d_i$ is a function $\N \mapsto \N$ where $d_i(t) = j$ indicates that
there are $j$ robots in node $u_i$ at instant $t$. If $d_i(t) = 0$,
then we say that $u_i$ is {\em free} at instant $t$, otherwise $u_i$
is said {\em occupied} at instant $t$. If $d_i(t) > 1$, then we say
that $u_i$ contains a {\em tower (of $d_i(t)$ robots)} at instant $t$.

As previously, to be as general as possible, our impossibility results
are written assuming the strong multiplicity detection, while our
algorithms assume the weak multiplicity detection.

\paragraph{Configurations and Views.}
To define the notion of {\em configuration}, we need to use an
arbitrary order $\prec$ on nodes. The system being anonymous, robots
do not know this order. (Actually, this order is used in the reasoning
only.) Let $v_1, \ldots, v_n$ be the list of the nodes in $G$ ordered
by $\prec$. The configuration at time $t$ is $d_1(t), \ldots,
d_{n}(t)$.
We denote by {\em initial configurations} the configurations from
which the system can start at time 0.  Every configuration where all
robots stay idle forever is said to be {\em terminal}.  Two
configurations $d_1, \ldots, d_{n}$ and $d_1', \ldots, d_{n}'$ are
{\em indistinguishable} ({\em distinguishable} otherwise) if and only
if there exists an automorphism $f$ on $G$ satisfying the additional
condition: $\forall v_i \in V$, we have $d_i = d_j'$ where $v_j =
f(v_i)$.

The {\em view} of robot $\mathcal R$ at time $t$ is a labelled graph
isomorphic to $G$, where every node $u_i$ is labelled by $d_i(t)$,
except the node where $\mathcal R$ is currently located, this latter
node $u_j$ is labelled by $d_j(t),*$. (Indeed, any robot knows the
multiplicity of the node where it is located.) Hence, from its view, a
robot can compute the view of all other robots, and decide whether
some other robots have the same view as its own.

Every decision to move is based on the view obtained during the last
Look action. However, it may happen that some edges incident to a node
$v$ currently occupied by the deciding robot look identical in its
view, {\em i.e.}, $v$ lies on a symmetric axis of the
configuration. In this case, if the robot decides to take one of these
edges, it may take any of them. As in related work ({\em
  e.g.},~\cite{FIPS07c,FIPS08c,LPT10c}), we assume the worst-case
decision in such cases, {\em i.e.} the actual edge among the
identically looking ones is chosen by the adversary.

\paragraph{Execution.}
We model the executions of our protocol in $G$ by the list of
configurations through which the system goes. So, an {\em execution}
is a maximal list of configurations $\gamma_0, \ldots,
\gamma_i$ such that $\forall j>0$, we have:
\begin{enumerate}
\item $\gamma_{j-1} \neq \gamma_j$.
\item $\gamma_j$ is obtained from $\gamma_{j-1}$ after some robots
  move from their locations in $\gamma_{j-1}$ to a neighboring node.
\item For every robot $\mathcal R$ that moves between $\gamma_{j-1}$
  and $\gamma_j$, there exists $0\leq j'\leq j$, such that $\mathcal
  R$ takes its decision to move according to its program and its view
  in $\gamma_{j'}$.
\end{enumerate}
An execution $\gamma_0, \ldots, \gamma_i$ is said to be {\em
  sequential} if and only if $\forall j>0$, exactly one robot moves
between $\gamma_{j-1}$ and $\gamma_j$.

\paragraph{Exploration.}
We consider the {\em exploration} problem, where $k$ robots, initially
placed at different nodes, collectively explore an $(i,j)$-grid before
stopping moving forever. By ``collectively'' explore we mean that
every node is eventually visited by at least one robot.  More
formally, a protocol $\mathcal P$ \emph{deterministically}
(resp. \emph{probabilistically}) solves the exploration problem if and
only if every execution $e$ of $\mathcal P$ starting from a {\em
  towerless} configuration satisfies: (1) $e$ terminates in
\emph{finite time} (resp. in \emph{finite expected time}), and (2)
every node is visited by at least one robot during $e$.

Observe that the exploration problem is not defined for $k > n$ and is
straightforward for $k=n$. (In this latter case the exploration is
already accomplished in the initial towerless configuration.)

\section{Bounds}\label{Bounds}

In this section, we first show that, except for trivial cases where
$k=n$, when robots are oblivious, the model is atomic, and the
multiplicity is strong, at least three robots are necessary to solve
the (probabilistic or deterministic) exploration of any grid (Theorem
\ref{theo:non2}). Moreover, in a $(2,2)$-Grid, four robots are
necessary (Theorem \ref{theo:22}). Finally, at least five robots are
necessary to solve the deterministic exploration of a $(3,3)$-Grid
(Theorem \ref{theo:imp4}).  In the two next sections, we show that all
these bounds are also sufficient to solve the deterministic
exploration in the asynchronous and non-atomic CORDA model.

Given that robots are oblivious and there are more nodes than robots,
any terminal configuration should be distinguishable from any possible
initial (towerless) configuration. So, we have:

\begin{rem}\label{rem:1}
  Any terminal configuration of any (probabilistic or deterministic)
  exploration protocol for a grid of $n$ nodes using $k<n$ oblivious
  robots contains at least one tower.
\end{rem}

\begin{theo}\label{theo:non2}
  There exists no (probabilistic or deterministic) exploration
  protocol in ATOM using $k \leq 2$ oblivious robots for any
  $(i,j)$-Grid made of at least $3$ nodes.
\end{theo}
\begin{proof}
  By Remark \ref{rem:1}, there is no protocol allowing one robot to
  explore any $(i,j)$-Grid made of at least 2 nodes. Indeed, any
  configuration is towerless in this case.  Assume by contradiction,
  that there exists a protocol $\mathcal P$ in ATOM to explore with 2
  oblivious robots an $(i,j)$-Grid made of at least 3 nodes.  Consider
  a sequential execution $e$ of $\mathcal P$ that terminates. (By
  definition, if $\mathcal P$ is deterministic, all its executions
  terminates; while if $\mathcal P$ is probabilistic, at least one of
  its sequential execution must terminate.)  Then, $e$ starts from a
  towerless configuration (by definition) and eventually reaches a
  terminal configuration containing a tower (by Remark
  \ref{rem:1}). As $e$ is sequential, the two last configurations of
  $e$ consist of a towerless configuration followed by a configuration
  containing one tower. These two configurations form a possible
  sequential execution that terminates while only two nodes are
  visited, thus a contradiction.
\end{proof}

Any $(2,2)$-Grid is isomorphic to a 4-size ring. It is shown in
\cite{DPT09c} that no (probabilistic or deterministic) exploration
using less than four oblivious robots is possible for any ring of size
at least four in ATOM. So:

\begin{theo}[\cite{DPT09c}]\label{theo:22}
  There exists no (probabilistic or deterministic) exploration
  protocol using $k \leq 3$ oblivious robots in ATOM for a
  $(2,2)$-Grid.
\end{theo}

\begin{lem}\label{lem:distinct}
  Considering any deterministic exploration protocol $\mathcal P$ in
  ATOM using $k$ oblivious robots for a $(3,3)$-Grid,
  there exist sequential executions of $\mathcal P$, $e = \gamma_0,
  \ldots, \gamma_w$, in which:
\begin{itemize}
\item For every $x,y$ with $0\leq x<y$, $\gamma_x$ and $\gamma_y$ are
  distinguishable.
\item Only the first configuration $\gamma_0$ is towerless.
\end{itemize}
\end{lem}
\begin{proof}
  Consider any exploration protocol $\mathcal P$ in ATOM 
  using $k$ oblivious robots for a $(3,3)$-Grid. Consider any
  sequential execution $e$ of $\mathcal P$. By definition of the
  exploration, $e$ is finite and starts from a towerless
  configuration. Moreover, the terminal configuration of $e$ contains
  a tower, by Remark \ref{rem:1}.

  Take the last towerless configuration of $e$ and all remaining
  configurations that follow in $e$ (all of them contain a tower) and
  form $e'$. $e'$ is a possible sequential execution of $\mathcal P$
  where only the first configuration is towerless.

  Let $e'=\alpha^0,\ldots,\alpha^m$. Let two configurations
  $\alpha^x =d^x_1, \ldots, d^x_{n}$ and $\alpha^y =
  d^y_1,\ldots,d^y_{n}$ of $e'$, that are indistinguishable with $0
  \leq x<y$. Then, by definition, there exists an automorphism $f$ on
  $G$ satisfying the additional condition: Let $v_0,\ldots,v_r$ be the
  nodes of $V$, for all $s \in [0..r]$, we have $d^x_s = d^y_{\ell}$
  where $v_{\ell} = f(v_s)$. Then,
  $\alpha^0,\ldots,\alpha^x,\beta^{y+1},\beta^{m}$ is a possible
  sequential execution of $\mathcal P$ such that $\forall z \geq y+1$,
  we have $\beta^{z}= d^{z}_{g(1)},\ldots,d^{z}_{g(n)}$ where $g$ is a
  bijection such that $\forall s \in [1..n]$, $f(v_s) = v_{g(s)}$ and
  $\alpha^{z}= d^{z}_1,\ldots,d^{z}_{n}$. Moreover, in
  $\alpha^0,\ldots,\alpha^x,\beta^{y+1},\beta^{m}$, the number of
  configurations indistinguishable from $\alpha^x$ decreases by
  one. Repeating the same construction, we eventually obtain a
  possible sequential execution $e''= \rho_0,\ldots,\rho_w$ of
  $\mathcal P$ starting from a towerless configuration only followed
  by configurations containing at least one tower such that for every
  $x,y$ with $0\leq x<y$, $\rho_x$ and $\rho_y$ are
  distinguishable.
\end{proof}

\begin{figure} 
  \begin{center}
  \epsfig{figure=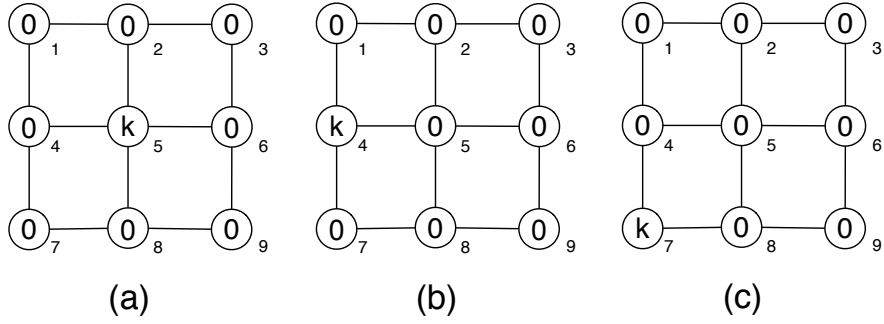,width=12cm}\\ 
  \end{center}
  \caption{Three possible configurations in a $(3,3)$-Grid with a
    tower of $k$ robots.\label{fig:tourK}}
\end{figure}

\begin{lem}\label{lem:tour3}
  Considering any deterministic exploration protocol $\mathcal P$ in ATOM
  model using $k$ oblivious robots for a $(3,3)$-Grid, if there exists
  an execution of $\mathcal P$ $e=\gamma_0\ldots\gamma_x\ldots$ where
  $\gamma_x$ contains a tower of $k$ robots, then there exists an
  execution $e'$ starting with the prefix $e=\gamma_0\ldots\gamma_x$
  such that at most one new node can be visited after $\gamma_x$.
\end{lem}
\begin{proof}
  Assume the existence of an execution of $\mathcal P$
  $e=\gamma_0\ldots\gamma_x\ldots$ where $\gamma_x$ contains a tower
  of $k$ robots. Then, $\gamma_x$ is not $\gamma_0$ and is
  indistinguishable from configuration $(a)$, $(b)$, or $(c)$ of
  Figure~\ref{fig:tourK}.  In Figure \ref{fig:tourK}, symbols inside
  the circles represent the multiplicity of the node and numbers next
  the circle are node's labels to help explanations only.  Without
  loss of generality, assume that $\gamma_x$ is either configuration
  $(a)$, $(b)$, or $(c)$.

  To visit a new node, one of the robots should eventually decide to
  move. Moreover, in $\gamma_x$, all robots have the same view. So,
  the adversary can choose any of them to move.
  \begin{itemize}
  \item[(1)] Consider configuration $(a)$. Then, all possible
    destinations for the robots are symmetric. So, the adversary can
    activate the robots in a way we retrieve configuration
    $\gamma_{x-1}$. Then, it can activate robots in a way that the
    system return to $\gamma_x$, and so on. Hence, in this case, there
    exists a possible execution of $\mathcal P$ that is infinite, a
    contradiction. So, from $(a)$, $\mathcal P$ cannot try to visit a new
    node.
  \item[(2)] Consider configuration $(b)$. 

    If robots synchronously move to
    node 5, node 5 may be unvisited. So, it is possible to visit a new
    node, but then we retrieve Case $(1)$. So, we can conclude
    that in this case from $(b)$ only one new node can be visited.
   
    If robots synchronously move to node 1 (resp. 7), then this node
    may be unvisited. So, it is possible to visit a new node. But, in
    node 1, all possible destinations for the robots are
    symmetric. So, the adversary can activate the robots in a way that
    we retrieve the previous configuration, if we want to visit
    another node. So, as for Case $(1)$, we can conclude that no new
    node can be visited, that is from $(b)$ only one new node can be
    visited.

  \item[(3)] Using a reasoning similar to case $(1)$, we can
    conclude that from $(c)$, $\mathcal P$ cannot try to visit a new node.
\end{itemize}
\end{proof}

\newcommand{\href}[1]{{\tt #1}}

\begin{lem}\label{lem:pascal}
  Assume that there exists a deterministic exploration protocol $\mathcal P$ in 
  ATOM model using $3$ oblivious robots for a
  $(3,3)$-Grid. Consider any suffix $\gamma_w,\ldots,\gamma_z$ of any
  sequential execution of $\mathcal P$ where:
\begin{itemize}
\item For every $x,y$ with $0\leq x<y$, $\gamma_x$ and $\gamma_y$ are
  distinguishable.
\item $\gamma_w$ contains a tower of $2$ robots.
\end{itemize}
Then, at most 4 new nodes can be visited from $\gamma_w$ before a
robot of the tower moves.
\end{lem}
\begin{proof}
  Proving this lemma is particularly tedious and error-prone because
  many cases must be taken into account (positions of robots, symmetry
  classes, etc.). The proof was thus completed as automatically as
  possible, by using model-checking techniques. The method is briefly
  sketched here, a detailed presentation, together with the source
  code and the necessary tools can be found on the
  web~\footnote{~\url{http://www-verimag.imag.fr/~raymond/misc/robots/}.}.  
  First, an operational model of the problem is built: this model is a
  reactive program that manages an abstract view of the grid and
  robots, according to a flow of (random) move commands. This model is
  restricted to the configurations relevant for the property: an
  immobile two-robots tower and a mobile single robot. The reactive
  program ({\em i.e.}, the model) computes the consequences of the
  moves induced by the input commands; in particular, it takes trace
  of the {\em visited} nodes, and the encountered indistinguishable
  configuration classes. As soon as such a class has been reached
  twice, a flag {\em stuck} is raised. And, all along the execution, a
  {\em validity} flag is computed that way: {\em stuck} $\Rightarrow$
  number of new {\em visited} nodes is $\leq$ 4. A model-checker tool
  is then used to check the following invariant: whatever be a
  sequence of input move commands, {\em valid} remains true. In other
  terms, the invariance of {\em valid} is sufficient to establish
  that, starting from any configuration with a tower and a single
  moving robot, at most 4 new nodes can be visited before the
  configuration becomes indistinguishable from some already
  encountered configuration. Concretely, the model is written in the
  Lustre language~\cite{HCRP91j,R08b}, and is itself partially
  generated by a "meta" program written in oCaml (which computes, in
  particular, the classes). The source is made of approximately 150
  lines of oCaml, and 100 lines of Lustre. The invariance checking is
  performed by the model-checker from the lustre distribution.
\end{proof}

\begin{theo}\label{theo:imp3}
  There exists no deterministic exploration protocol in ATOM using
  $k\leq 3$ oblivious robots for a $(3,3)$-Grid.
\end{theo}
\begin{proof}
  According to Theorem \ref{theo:non2}, we only need to consider the
  case of $3$ robots.

  Assume that there exists an exploration protocol $\mathcal P$ in 
  ATOM for a $(3,3)$-Grid using 3 robots.  By Lemma
  \ref{lem:distinct}, there exists a sequential execution $e =
  \gamma_0, \ldots, \gamma_w$ that starts from a towerless
  configuration, only followed by configurations containing at least
  one towers, and such that for every $x,y$ with $0\leq x<y$,
  $\gamma_x$ and $\gamma_y$ are distinguishable.

  In $\gamma_0$, 3 nodes are visited. The execution being sequential,
  no new node is visited in the first step where a tower of two robots
  is created. So, in $\gamma_1$, 3 nodes are visited and there exists
  a tower of two robots $\mathcal R_1$ and $\mathcal R_2$.
\begin{itemize}
\item Assume that $\mathcal R_1$ and $\mathcal R_2$ never moved after
  $\gamma_1$. Then, by Lemma \ref{lem:pascal}, at most 4 new nodes are
  visited until the termination of $e$. So, at the termination of $e$,
  at most 7 distinct nodes have been visited, a contradiction.
\item Assume that $\mathcal R_1$ or $\mathcal R_2$ eventually
  moved. Let $\gamma_{\ell}$ the first configuration from which $\mathcal
  R_1$ or $\mathcal R_2$ moves. From the previous case, at most 7
  distinct nodes have been visited before $\gamma_{\ell}$. The execution
  being sequential, only one robot of the tower moves during the step
  from $\gamma_{\ell}$ to $\gamma_{i+1}$ and as in $e$ only the first
  configuration is towerless, that robot moves to an occupied
  node. Now, the view of $\mathcal R_1$ and $\mathcal R_2$ are
  identical in $\gamma_{\ell}$. So, there exists an execution $e'$ starting
  from the prefix $\gamma_0, \ldots, \gamma_{\ell}$ where both $\mathcal
  R_1$ and $\mathcal R_2$ move from $\gamma_{\ell}$ to the same occupied
  node. As no new node is visited during the step, still at most 7
  nodes are visited once the system is in the new configuration and
  this configuration contains a tower of 3 robots. By Lemma
  \ref{lem:tour3}, at most one new node is visited from this latter
  configuration. So, at the termination of $e'$, at most 8 distinct
  nodes have been visited, a contradiction.
\end{itemize}
\end{proof} 

\begin{figure} 
  \begin{center}
  \epsfig{figure=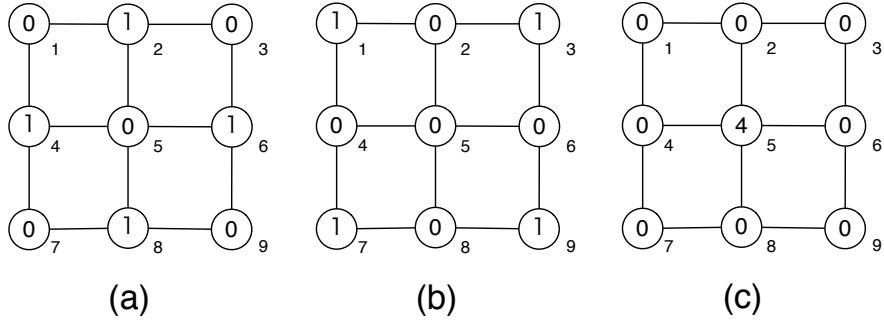,width=12cm}\\ 
  \end{center}
  \caption{Three possible configurations in a $(3,3)$-Grid with $4$
    robots. Numbers inside the circles represent the multiplicity
    of the node. Numbers near the circles are node's labels that are used 
   to ease the explanations only. \label{fig:4RG33}}
\end{figure}

\begin{theo}\label{theo:imp4}
   There exists no deterministic exploration protocol in ATOM using
  $k\leq 4$ oblivious robots for a $(3,3)$-Grid. 
\end{theo}
\begin{proof}
 According to Theorem \ref{theo:imp3}, we only need to consider the
  case of $4$ robots.

  Assume, by the way of contradiction, that there exists an
  exploration protocol $\mathcal P$ for a $(3,3)$-Grid with 4 robots
  in ATOM.

  Figure \ref{fig:4RG33} depicts three possible configurations for a
  $(3,3)$-Grid with 4 robots. In Figure \ref{fig:4RG33}, symbols
  inside the circles represent the multiplicity of the node and numbers
  next the circle are node's labels to help explanations only. Note
  that both Configuration $(a)$ and $(b)$ can be initial configuration.

  From now on, consider any synchronous execution of $\mathcal P$
  (synchronous executions are possible in the asynchronous model)
  starting from configuration $(a)$. By ``synchronous'' we mean that
  robots execute each operation of each cycle at the same time.

  Configuration $(a)$ is not a terminal configuration by Remark
  \ref{rem:1}. So at least one robot move in the next Move
  operation. Moreover, the views of all robots are identical in
  $(a)$. So, every robot moves in the next Move operation. Two cases
  are possible:
\begin{itemize}
\item Every robot moves to Node 5 and the system reaches Configuration
  $(c)$.  In this case, none of the corners has been visited, so
  Configuration $(c)$ is not terminal and at least one robot moves in
  during the next Move operation. Moreover, the views of all robots
  are identical, so every robot moves in the next Move operation. Each
  robot cannot differentiate its four possible possible
  destinations. So, the adversary can choose destinations so that
  the system reaches configuration $(a)$ again.
\item Every robot moves to a corner node and as its view is symmetric,
  the destination corner is chosen be the adversary. In this case, the
  adversary can choose destinations so that the system reaches
  configuration $(b)$. Configuration $(b)$ being not terminal, at
  least one robot moves in during the next Move operation. Moreover,
  the views of all robots are identical, so every robot moves in the
  next Move operation. Each robot cannot differentiate its two
  possible possible destinations. So, the adversary can choose to
  destinations so that the system reaches configuration $(a)$ again.
\end{itemize}
From the two previous case, we can deduce that there exist executions
of $\mathcal P$ that never terminates, so $\mathcal P$ is not an
exploration protocol, a contradiction.
\end{proof}

\section{Deterministic solution using three robots}\label{sec:3R}

In this section, we focus on solutions for the exploration problem
that use three robots only, in CORDA, and assuming weak multiplicity
detection. Recall that there exists no deterministic solution for the
exploration using three robots in a $(2,2)$- or $(3,3)$-grid assuming
that model (Section~\ref{Bounds}). Moreover, exploring a $(3,1)$-grid
using three robots is straightforward.  So, we consider all remaining
cases. We split our study in two cases. A general deterministic
solution for any $(i,j)$-grid such that $j>3$ is given in Subsection
\ref{3-gen}. The particular case of the $(2,3)$-grid is solved in
Subsection~\ref{2-3}.

\subsection{General Solution}\label{3-gen}

\paragraph{Overview.}

Our deterministic algorithm works according to the following three
main phases:
\begin{description}
\item [\noindent {\tt Set-Up} phase:] The aim of this phase is to
  create a single line of robots starting at a corner and along one of
  the longest borderlines of the grid --- refer to
  Figure~\ref{line}. Let us refer to this configuration as the {\tt
    Set-Up} configuration. The phase can be initiated from any
  arbitrary towerless configuration that is not a {\tt Set-Up}
  configuration.  Note that no tower is created during this phase.
\item [\noindent{\tt Orientation} phase:] This phase follows the {\tt
    Set-Up} phase.  Starting from a {\tt Set-Up} configuration, this
  phase aims at giving an orientation to the grid. To achieve that,
  one tower is created allowing the robots to establish a common
  coordinate system --- refer to Figure~\ref{EXPP}.  The resulting
  configuration is called an {\tt Oriented} configuration.
\item [\noindent{\tt Exploration} phase:] This phase starts from an
  {\tt Oriented} configuration in which exactly one node is occupied
  by one single robot, called {\em Explorer}.  Based on the coordinate
  system defined during the {\tt Orientation} phase, the explorer
  visits all the nodes, except three already visited ones --- refer to
  Figure~ \ref{Exploration}, page \pageref{Exploration}.
\end{description}
We now describe the three above phases in more details.

\begin{figure} 
 \begin{minipage}[b]{.46\linewidth}
  \begin{center}
  \epsfig{figure=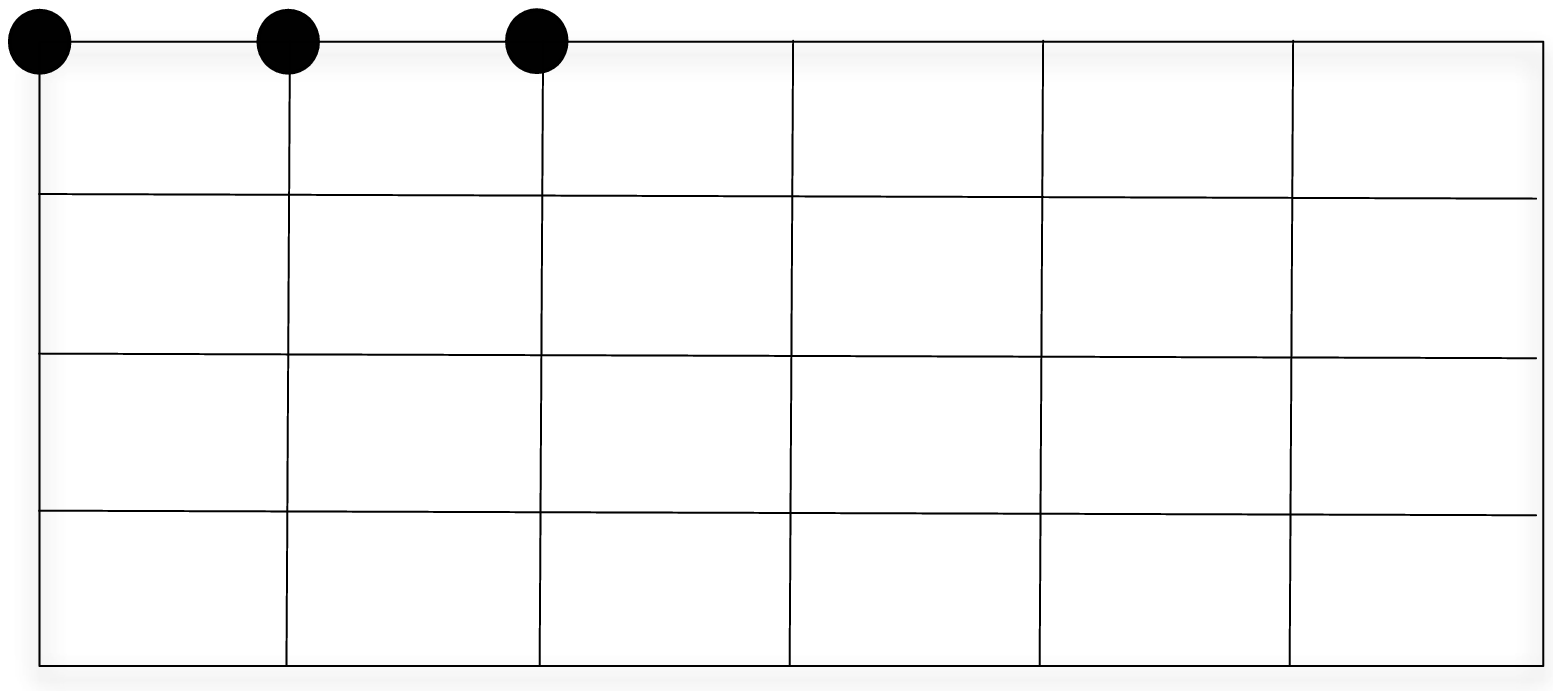,width=5cm}\\
  \caption{{\tt Set-Up} Configuration}\label{line}
  \end{center}
 \end{minipage} \hfill
\begin{minipage}[b]{.46\linewidth}
\begin{center}
 \epsfig{figure=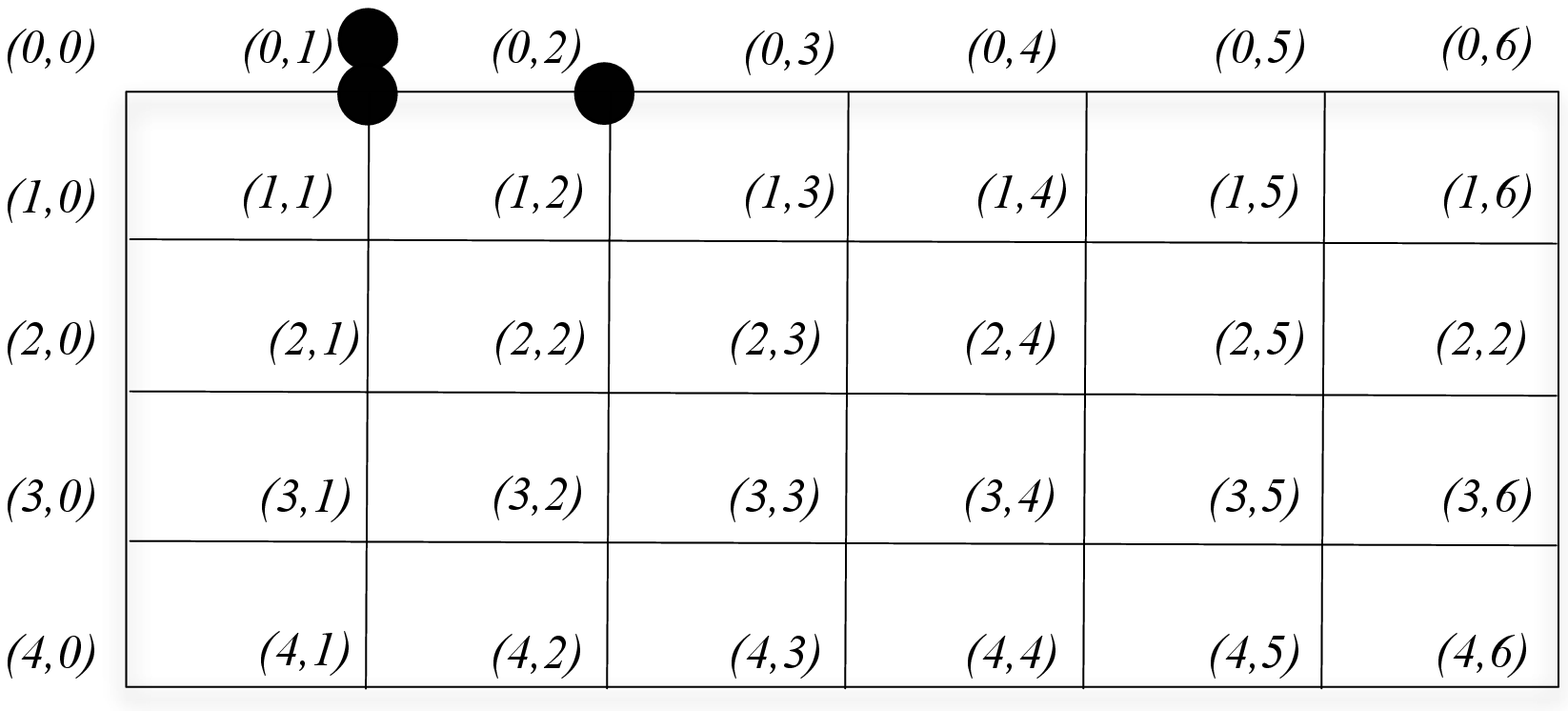,width=6cm}\\
 \caption{Coordinate system built by the {\tt Orientation}
   phase}\label{EXPP}
  \end{center}
 \end{minipage}\hfill
\end{figure}

\paragraph{Set-Up Phase.}
Starting from any towerless configuration, the {\tt Set-Up} phase ends
in a {\tt Set-Up} configuration, where there is a single line of
robots starting at a corner and along a longest borderline of the
grid.  In this phase, we distinguish three main configurations:
\begin{description}
\item[{\tt Leader}:] In such a configuration, there is exactly one
  robot located at a corner of the grid.
\item[\noindent {\tt Choice}:] In such a configuration, at least two
  robots are located at a corner of the grid. We choose one of them to
  remain at a corner. The other ones have to leave their corner.
\item[\noindent {\tt Undefined}:] In such a configuration, there is no
  robot at any corner of the grid. The idea is then to elect one robot
  that will move to join a corner of the grid.
\end{description}

In the following, we present the behavior of the three robots,
respectively referred to as $\mathcal R1, \mathcal R2$, and $\mathcal
R3$,\footnote{Recall that robots are anonymous, so these notations are
  used to ease the explanations only.} in each of the main
configurations. These configurations are declined into several
subconfigurations.
      \begin{enumerate}

      \item The configuration is of type {\tt Leader}: In such a
        configuration, there is exactly one robot that is at a corner
        of the grid. Let $\mathcal R1$ be this robot. We consider the
        following subcases:

        \noindent A) The configuration is of type {\tt Strict-Leader}:
        In such a configuration, there is no other robot on any
        borderline having the corner where $\mathcal R1$ is located as
        extremity. In this case, the robots that are the closest to
        $\mathcal R1$ are the ones allowed to move. Their destination
        is their adjacent free node on a shortest path towards the
        closest free node that is on a longest borderline having the
        corner where $\mathcal R1$ is located as extremity. (If there
        is several shortest paths, the adversary makes the choice.)

        \noindent B) The configuration is of type {\tt Half-Leader}:
        In such a configuration, among $\mathcal R2$ and $\mathcal
        R3$, only one robot, say $\mathcal R2$, is on a borderline
        having the corner where $\mathcal R1$ is located as
        extremity. Two subcases are possible:
                                \begin{itemize}
                                \item The configuration is of type
                                  {\tt Half-Leader1}: $\mathcal R2$ is
                                  on a longest borderline. In this
                                  case, the third robot $\mathcal R3$
                                  is the one allowed to move. Its
                                  destination is an adjacent free
                                  node towards a closest free node
                                  on the borderline that contains both
                                  $\mathcal R1$ and $\mathcal R2$.
                                  (If there is several shortest paths,
                                  the adversary makes the choice.)
                                \item The configuration is of type
                                  {\tt Half-Leader2}: $\mathcal R2$ is
                                  not on the longest borderline. In
                                  this case, $\mathcal R2$ is the one
                                  allowed to move, its destination is
                                  the adjacent free node outside the
                                  borderline, if any. In the case
                                  where there is no such a free node,
                                  $\mathcal R2$ moves to a free node
                                  on its own borderline (In case of
                                  symmetry, the adversary makes the
                                  choice.)
                                 \end{itemize}

                                 \noindent C) The configuration is of
                                 type {\tt All-Leader}: All the robots
                                 are on a borderline having the corner
                                 where $\mathcal R1$ is located as
                                 extremity. In this case, $\mathcal
                                 R2$ and $\mathcal R3$ are not
                                 necessary on the same
                                 borderline. Thus, we have two
                                 subcases:
                               \begin{itemize}
                               
                               \item The configuration is of type {\tt
                                   Fully-Leader}: In such a
                                 configuration, all the robots are on
                                 the same borderline, $D1$. The two
                                 following subcases are then possible:
                                 
                                 \noindent $(i)$ The configuration is
                                 of type {\tt Fully-Leader1}: In this
                                 case, $D1$ is a longest
                                 borderline. If the robots form a
                                 line, then the {\tt Set-Up}
                                 configuration is reached and the
                                 phase is done. Otherwise, let
                                 $\mathcal R2$ be the closest robot
                                 from $\mathcal R1$. If $\mathcal R1$
                                 and $\mathcal R2$ are not neighbors,
                                 then $\mathcal R2$ is the only
                                 allowed to move and its destination
                                 is the adjacent free node towards
                                 $\mathcal R1$. In the other case,
                                 $\mathcal R3$ is the only robot
                                 allowed to move and its destination
                                 is the adjacent free node towards
                                 $\mathcal R2$.
                                       
                                 \noindent $(ii)$ The configuration is
                                 of type {\tt Fully-Leader2}: In this
                                 case, $D1$ is not the longest
                                 borderline. Then, the robot among
                                 $\mathcal R2$ and $\mathcal R3$ that
                                 is the closest to $\mathcal R1$
                                 leaves the borderline by moving to
                                 its neighboring free node outside the
                                 borderline.
                               \item The configuration is of type {\tt
                                   Semi-Leader}: $\mathcal R2$ and
                                 $\mathcal R3$ are not on the same
                                 borderline. Two subcases are
                                 possible:

                                 \noindent $(i)$ The configuration is of
                                 type {\tt Semi-Leader1}: In this
                                 case, $i \ne j$. The unique robot
                                 among $\mathcal R2$ and $\mathcal R3$
                                 which is located on the smallest
                                 borderline moves to the adjacent free
                                 node outside its borderline.

                                 \noindent $(ii)$ The configuration is
                                 of type {\tt Semi-Leader2}: In this
                                 case, $i=j$. Let denote by
                                 $Dist(\mathcal R,\mathcal R')$ the
                                 {\em distance} (that is, the length
                                 of a shortest path) in the grid
                                 between the two nodes where $\mathcal
                                 R$ and $\mathcal R'$ are respectively
                                 located.  If $Dist(\mathcal
                                 R1,\mathcal R2) \ne Dist(\mathcal
                                 R1,\mathcal R3)$ then the robot among
                                 $\mathcal R2$ and $\mathcal R3$ that
                                 is the closest to $\mathcal R1$ is
                                 the only one allowed to move, its
                                 destination is the adjacent free node
                                 outside the borderline. Otherwise
                                 ($Dist(\mathcal R1,\mathcal R2)$ $=$
                                 $Dist(\mathcal R1,\mathcal R3)$),
                                 either (a) there is a free node
                                 between $\mathcal R1$ and $\mathcal
                                 R2$, or (b) $\mathcal R1$ is both
                                 neighbor of $\mathcal R2$ and
                                 $\mathcal R3$.  In case (a),
                                 $\mathcal R1$ is the only robot
                                 allowed to move and its destination
                                 is the adjacent free node towards one
                                 of its two borderlines (the adversary
                                 makes the choice). In case (b),
                                 $\mathcal R2$ and $\mathcal R3$ move
                                 and their destination is their
                                 adjacent free node on their
                                 borderline.
                               \end{itemize}                                         
                             \item The configuration is of type {\tt
                                 Choice}: At least two robots are
                               located at a corner. We consider two
                               cases:
               
                               \noindent A) The configuration is of
                               type {\tt Choice1}: In this
                               configuration, there are exactly two
                               robots that are located at a corner of
                               the grid. Let $\mathcal R1$ and
                               $\mathcal R2$ be these robots.
                               
                               \begin{itemize}
                               \item In the case where $\mathcal R3$
                                 is on the same borderline as either
                                 $\mathcal R1$ or $\mathcal R2$ but
                                 not both --- suppose $\mathcal R1$
                                 --- then $\mathcal R2$ is the one
                                 allowed to move, its destination is
                                 the adjacent free node towards the
                                 closest free node of the borderline
                                 that contains both $\mathcal R1$ and
                                 $\mathcal R3$.

                             \item In the case where the three robots
                               are on the same borderline. Then:

                               \noindent $(i)$ If $Dist(\mathcal
                               R1,\mathcal R3)$ $\neq$ $Dist(\mathcal
                               R2,\mathcal R3)$, then the robot among
                               $\mathcal R1$ and $\mathcal R2$ that is
                               farthest to $\mathcal R3$ moves to the
                               adjacent free node on the borderline
                               towards $\mathcal R3$.

                               \noindent $(ii)$ Otherwise
                               ($Dist(\mathcal R1,\mathcal R3)$ $=$
                               $Dist(\mathcal R2,\mathcal R3)$), and
                               $\mathcal R3$ has either or not an
                               adjacent free node on the
                               borderline. In the former case,
                               $\mathcal R3$ moves to an adjacent
                               free node on the borderline towards
                               either $\mathcal R1$ or $\mathcal R2$
                               (the adversary makes the choice). In
                               the latter case, $\mathcal R3$ moves to
                               its adjacent free node outside the
                               borderline.

                             \item If $\mathcal R3$ is not on any
                               borderline, it moves to an adjacent
                               free node on a shortest path towards
                               the closest free node that is on a
                               longest borderline that contains either
                               $\mathcal R1$ or $\mathcal R2$.  (In
                               case of symmetry, the adversary makes
                               the choice.)
                               \end{itemize}
                               \noindent B) The configuration is of
                               type {\tt Choice2}: In this
                               configuration, all the robots are
                               located at a corner. The robot allowed
                               to move is the one that is located at a
                               node that is common to the two
                               borderlines of the other robots. Let
                               $\mathcal R1$ be this robot. The
                               destination of $\mathcal R1$ is the
                               adjacent free node on a longest
                               borderline. (In case of symmetry, the
                               adversary makes the choice.)
                  
                             \item The configuration is of type {\tt
                                 Undefined}: In this configuration,
                               there is no robot that is located at
                               any corner. The cases below are then
                               possible:
                  
                               \noindent A) The configuration is of
                               type {\tt Undefined1}: In this case,
                               $i=j$ and there is one borderline that
                               contains two robots $\mathcal R1$ and
                               $\mathcal R2$ such that $\mathcal R1$
                               is closer from a corner than $\mathcal
                               R2$ and $\mathcal R3$. Let $D1$ be this
                               borderline. $\mathcal R3$ is the only
                               one allowed to move and its destination
                               is an adjacent free node on a shortest
                               path towards a closest free node of
                               $D1$.  (If there are several shortest
                               paths, the adversary makes the choice.)

                               \noindent B) The configuration is of
                               type {\tt Undefined2}: It is any
                               configuration different from {\tt
                                 Undefined1}, where there is exactly
                               one robot that is the closest to a
                               corner. In this case, this robot is the
                               only one allowed to move, its
                               destination is an adjacent free node
                               on a shortest path to a closest
                               corner. (If there are several
                               possibilities, the adversary makes the
                               choice.)

                               \noindent C) The configuration is of
                               type {\tt Undefined3}: There are
                               exactly two robots that are closest to
                               a corner. Let $\mathcal R1$ and
                               $\mathcal R2$ be these two robots.

                               \begin{itemize}
                               \item If $Dist(\mathcal R1,\mathcal R3)
                                 = Dist(\mathcal R2,\mathcal R3)$ then
                                 $\mathcal R3$ is the only one allowed
                                 to move, and either $Dist(\mathcal
                                 R1,$ $\mathcal R3)$ $=$ $1$ or
                                 $Dist(\mathcal R1,\mathcal R3) > 1$.
                                 In the former case, $\mathcal R3$
                                 moves to an adjacent free node. (If
                                 there are two possibilities, the
                                 adversary make the choice.) In the
                                 latter case, $\mathcal R3$ moves to
                                 an adjacent free node that is on a
                                 shortest path towards either
                                 $\mathcal R1$ or $\mathcal R2$ but
                                 not both.

                               \item If $Dist(\mathcal R1,\mathcal R3)
                                 \ne Dist(\mathcal R2,\mathcal R3)$
                                 then the robot among $\mathcal R1$
                                 and $\mathcal R2$ that is closest to
                                 $\mathcal R3$ is the only one allowed
                                 to move. Its destination is the
                                 adjacent free node that is on a
                                 shortest path to a closest
                                 corner. (If there are several
                                 possibilities, the adversary makes
                                 the choice.)
                               \end{itemize}
                  
                               \noindent D) The configuration is of
                               type {\tt Undefined4}: There are three
                               robots that are closest to a
                               corner. Again, four cases are possible:
                               \begin{itemize}
                               \item The configuration is of type {\tt
                                   Undefined4-1}: There is exactly one
                                 robot that is on a borderline. In
                                 this case, this robot is the only one
                                 allowed to move. Its destination is
                                 an adjacent free node that is on a
                                 shortest path to a closest
                                 corner. (In case of two shortest
                                 paths, the adversary breaks the
                                 symmetry in the first step.)
                               \item The configuration is of type {\tt
                                   Undefined4-2}: In such a
                                 configuration, there are exactly two
                                 robots on a borderline. Let $\mathcal
                                 R1$ and $\mathcal R2$ be these two
                                 robots. The robot allowed to move is
                                 $\mathcal R3$. Its destination is the
                                 adjacent free node towards a closest
                                 corner.  (The adversary may have to
                                 break the symmetry.)

                               \item The configuration is of type {\tt
                                   Undefined4-3}: The three robots are
                                 on borderlines of the grid. 

                                 \noindent $(i)$ If there are more than
                                 one robot on the same borderline. In
                                 this case, there are exactly two
                                 robots on the same borderline, and
                                 let $\mathcal R1$ and $\mathcal R2$
                                 be these robots. Then $\mathcal R3$
                                 is the only one allowed to move and
                                 its destination is an adjacent free
                                 node towards a closest corner. (The
                                 adversary may have to break the
                                 symmetry.)

                                 \noindent $(ii)$ If there is at most one
                                 robot on each borderline: Exactly one
                                 borderline is perpendicular to the
                                 two others. The robot on that
                                 borderline is the only one allowed to
                                 move and its destination is the
                                 adjacent node towards a closest
                                 corner.  (The adversary may have to
                                 break the symmetry.)

                               \item The configuration is of type {\tt
                                   Undefined4-4}: In this case, there
                                 is no robot on any borderline.  

                                 \noindent $(i)$ In the case where there
                                 are two robots, $\mathcal R1$ and
                                 $\mathcal R2$, that are closest to
                                 the same corner, and this corner is
                                 not a closest corner to $\mathcal
                                 R3$, then $\mathcal R3$ is the only
                                 robot allowed to move and its
                                 destination is an adjacent free node
                                 on a shortest path towards a closest
                                 corner. (If there are several
                                 possibilities, the adversary makes
                                 the choice.)

                                 \noindent $(ii)$ In the case where
                                 there are two robots, $\mathcal R1$
                                 and $\mathcal R2$, that are closest
                                 to corners $C1$ and $C2$,
                                 respectively, where $C1 \neq C2$, and
                                 $\mathcal R3$ is closest to both $C1$
                                 and $C2$, then $\mathcal R3$ is the
                                 only one allowed to move (refer to
                                 Figure \ref{Undef4-4}), and it moves
                                 toward $C1$ or $C2$ according to a
                                 choice of the adversary. 

                                 \noindent $(iii)$ In the case where all
                                 the robots are closest to different
                                 corners, there is one robot $\mathcal
                                 R1$ whom corner is between the two
                                 other targeted corners of $\mathcal
                                 R2$ and $\mathcal R3$. The robot
                                 allowed to move is $\mathcal R1$, its
                                 destination is an adjacent free node
                                 on a shortest path towards its
                                 closest corner. (If there are several
                                 shortest paths, the adversary makes
                                 the choice.)

                               \end{itemize}
                          \end{enumerate}

\begin{figure}
  \begin{center}
  \epsfig{figure=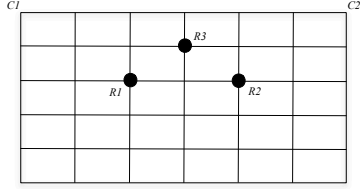,width=5cm}\\
  \caption{Sample of a configuration of type {\tt
      Undefined4-4}}\label{Undef4-4}
  \end{center}
 \end{figure}

 The correctness of the {\tt Set-Up} phase is established by Lemmas
 \ref{lem:notower} and \ref{lem:setup}.

\begin{lem}\label{lem:notower}
Starting from any arbitrary towerless configuration, {\tt Set-Up} phase does not create any tower.
\end{lem}
\begin{proof}
  It is clear that in the case where one robot is allowed to move, no
  tower is created because the robot always moves to an free adjacent
  node. Thus lets consider the cases in which there are at least two
  robots that are allowed to move:

\begin{itemize}

\item The configuration is of type {\tt Strict-Leader}: Suppose that
  the robot that is at the corner is $\mathcal R1$, and the two other ones
  (that are neither at a corner nor at the same borderline as $\mathcal R1$)
  are $\mathcal R2$ and $\mathcal R3$, respectively. $\mathcal R2$ and $\mathcal R3$ are allowed to move
  at the same time only in the case they are at the same distance from
  $\mathcal R1$. Since their destination is their adjacent free node on the
  shortest path towards the longest borderline that contains $\mathcal R1$, we
  are sure that the both will move to different free nodes. Thus no
  tower is created in this case.

\item The configuration is of type {\tt Semi-Leader2}: we consider the
  case in which $Dist(\mathcal R1, \mathcal R2) = Dist(\mathcal R1, \mathcal R3)$ such as there is
  no free node between $\mathcal R1$ and both $\mathcal R2$ and $\mathcal R3$ respectively. It
  is clear that if the adversary activates them at the same time no
  tower is created since they move to their adjacent free node on the
  borderline they belong to, in the opposite direction of $\mathcal R1$ (recall
  that they are in two different borderlines). In the case the
  adversary activates only one robot ($\mathcal R2$), no tower is created as
  well since it moves to its adjacent free node on the borderline it
  belongs to (note that is this case $i=j$). Note that the
  configuration reached remains of type Semi-leader2, however,
  $Dist(\mathcal R1, \mathcal R2) \ne Dist(\mathcal R1, \mathcal R3)$. Thus the robot that is
  allowed to move now is $\mathcal R3$, which is the one that was supposed to
  move at the first place. Thus either we retrieve the configuration
  in which both robots moved (this will happen in the case $\mathcal R3$ has an
  outdated view). Or the configuration reached is of type Half leader1
  and all the robots have a correct view.


\end{itemize}

From the cases above we can deduce that starting from any configuration that is towerless, {\tt Set-Up} phase does not create any tower and the lemma holds.
\end{proof}

Lemma \ref{lem:setup} is established using the following three
technical lemmas.

\begin{lem}\label{Leader}
  Starting from a configuration of type {\tt Leader}, a configuration
  of type {\tt Set-Up} is reached in a finite time.
\end{lem}   
\begin{proof}
  In a configuration of type {\tt Leader}, there is only one robot
  that is at the corner (suppose that this robot is
  $\mathcal R1$). It is easy to see that in the case $i\ne j$ all the robots
  will be on the longest borderline that contains $\mathcal R1$
  (refer to Strict Leader, HalfLeader1 configurations). Once the
  robots on the same longest borderline, it is also easy to create a
  line of robots keeping one robot at the corner. (The robot ($\mathcal R2$)
  that is the closest to $\mathcal R1$ moves first until it becomes neighbor of
  $\mathcal R1$. Once it is done, the remaining robot ($\mathcal R3$) moves to become
  neighbor of $\mathcal R2$.) Hence we are sure that a configuration of type
  {\tt Set-Up} is reached in a finite time. In the case $i=j$ when the
  robots move to the closest borderline that contains $\mathcal R1$ either we
  have the same result as when $i\ne j$ (all the robots will be on the
  same borderline) and hence we are sure to reach a configuration of
  type {\tt Set-Up}. Or, each robot $\mathcal R2$ and $\mathcal R3$ is on the same
  borderline as $\mathcal R1$, however both of them are on different
  borderlines. The sub-cases are then possible as follow:

\begin{enumerate}
\item $Dist (\mathcal R1,\mathcal R2) \ne Dist (\mathcal R1,\mathcal R3)$. In this case, the robot
  that is the closest to $\mathcal R1$ moves to its adjacent node outside its
  own borderline (Let this robot be $\mathcal R2$). Note that when it moves,
  its new destination is the closest free node on the same borderline
  as both $\mathcal R1$ and $\mathcal R3$ (see {\tt Semi-Leader2} configuration). Thus
  we are sure that $\mathcal R2$ will be on the same borderline of $\mathcal R1$ and
  $\mathcal R3$ in a finite time, thus we are sure that the {\tt Set-Up}
  configuration is reached in a finite time.
\item $Dist (\mathcal R1,\mathcal R2)=Dist (\mathcal R1,\mathcal R3)$. The two sub-case below are
  possible:

\begin{enumerate}
\item \label{EmpNode} There is an free node between $\mathcal R1$ and the
  other robots. $\mathcal R1$ is the one that will move, its destination is its
  adjacent free node on one of its two adjacent borderlines (Suppose that it moves towards $\mathcal R2$). Note
  that once it has moved, all the robots are in a borderline such as there
  is one borderline that contains two robots ($\mathcal R1$ and $\mathcal R2$), let $D1$ be this
  borderline (the configuration is of type {\tt Undefined1}). The
  robot allowed to move is $\mathcal R3$ (Note that $\mathcal R3$ is not part of $D1$), 
  its destination is its adjacent free node on a shortest path towards the closest free node of $D1$.  
  Once it moves, it becomes at the same distance as $\mathcal R1$ from a corner. The configuration becomes of type 
   {\tt Undefined3} such that $Dist (\mathcal R1,\mathcal R2)\ne Dist (\mathcal R1,\mathcal R3)$. $\mathcal R1$ is 
   the only one allowed to move, its destination is its adjacent empty node towards the corner. Once it moves, it joins
   one corner of the grid. The configuration becomes of type {\tt Semi-Leader2} such that $Dist (\mathcal R1,\mathcal R2)\ne Dist (\mathcal R1,\mathcal R3)$. $\mathcal R3$ is the only robot allowed to move, its
  destination is its adjacent free node outside the borderline it
  belongs to. Once it moves, its new destination will be the
  borderline that contains two robots. Thus, we are sure that all the
  robots will be part of the same borderline in a finite time. It is
  clear that from this configuration is easy to build a configuration
  of type {\tt Set-Up}. (Note that it is easy to break the symmetry,if
  any, since we have three robots.)

\item There is no free node between $\mathcal R1$ and the other robots $\mathcal R2$
  and $\mathcal R3$. In this case, $\mathcal R2$ and $\mathcal R3$ will be the ones allowed to
  move. Their destination is their adjacent free node on their
  borderline. In the case the adversary activates them at the same
  time, we retrieve case \ref{EmpNode}. If the adversary activates
  only one of the two robots, the configuration reached will be of
  type {\tt Semi-Leader2} such as $Dist (\mathcal R1,\mathcal R2) \ne Dist
  (\mathcal R1,\mathcal R3)$, thus, The robot that is the closest to $\mathcal R1$ is the one
  that is allowed to move. (Note that this robot is the one that was
  supposed to move at the first place.) If it has an outdated view it
  will move to its adjacent free node and we retrieve case
  \ref{EmpNode}. If not, it will move to its adjacent free node
  outside its borderline. When it does, its new destination is the closest
  free node on the same borderline of the two other robots. Note that
  when such a robot joins the new borderline, the configuration is of
  type {\tt Set-Up}.
\end{enumerate}
 \end{enumerate}
From the cases above, we can deduce that starting from a configuration of type {\tt Leader}, a configuration of type {\tt Set-Up} is reached in a finite time and the lemma holds.

\end{proof}

\begin{lem}\label{Choice}
Starting from a configuration of type {\tt Choice}, a configuration of type {\tt Leader} in reached in a finite time.
\end{lem}
\begin{proof}
  It is clear that in the case where all the robots are on one corner
  of the grid, the next configuration reached is of type {\tt Choice1}
  since there will be a single robot that will move (refer to
  Configuration of type {\tt Choice2}). Note that when the
  configuration is of type {\tt Choice1} the cases below are possible
  (Let the robots that are at the corner be $\mathcal R1$ and $\mathcal R2$
  respectively and the third robot be $\mathcal R3$):

\begin{enumerate}
\item $\mathcal R3$ is on the same borderline $D1$ as $\mathcal R1$ (Note that in this case $\mathcal R2$ is not on $D1$). In this case, $\mathcal R2$ is the one allowed to move. Note that once it moves, it leaves the corner and the configuration will be of type {\tt Leader} (refer to {\tt Choice1}, case (i)).
\item All the robots are on the same borderline $D1$. In this case, the robots $\mathcal R3$ will be used to elect one of the two robots at the corner (refer to {\tt Choice1} configuration case (ii)). If $Dist(\mathcal R1,\mathcal R3) \ne Dist(\mathcal R2,\mathcal R3)$ then the robot that is the farthest from $\mathcal R3$ leaves the corner, thus, the configuration will contain a single robot that is at one corner. Hence the configuration will be of type {\tt Leader} in a finite time. In the case $Dist(\mathcal R1,\mathcal R3)=Dist(\mathcal R2,\mathcal R3)$, ($a$) if there is at least one empty node between $\mathcal R1$ and $\mathcal R3$ then $\mathcal R3$ will be the one allowed to move on the borderline towards either $\mathcal R1$ or $\mathcal R2$ breaking the symmetry. Thus, we retrieve the case in which $Dist(\mathcal R1,\mathcal R3) \ne Dist(\mathcal R2,\mathcal R3)$.
($b$) In the case where there is no empty nodes between $\mathcal R1$ and $\mathcal R3$, then $\mathcal R3$ is the one allowed to move. Its destination is its adjacent node outside the borderline. Once it moves, it remains the only one allowed to move in the configuration reached. Its destination is its adjacent node on a shortest path towards the closest free node that is on a longest borderline that contains either $\mathcal R1$ or $\mathcal R2$ (the adversary makes the choice). Once it moves we retrieve the case in which $Dist(\mathcal R1,\mathcal R3) \ne Dist(\mathcal R2,\mathcal R3)$.
 Thus we are sure that a configuration of type {\tt Leader} is reached in a finite time.
\item $\mathcal R3$ is not on a borderline. In this case, $\mathcal R3$ is the one allowed to move. Its destination is its adjacent free node on a shortest path towards the closest longest borderline that contains either $\mathcal R1$ or $\mathcal R2$. Thus we are sure that one of the two cases described above will be reached (refer to {\tt Choice1} configuration, case (iii)).
\end{enumerate} 

From the cases above we can deduce that a configuration of type {\tt Leader} is reached in a finite time and the lemma holds.
\end{proof}

\begin{lem}\label{Undefined}
  Starting from a configuration of type {\tt Undefined}, a
  configuration of 
  type {\tt Leader} is
  reached in a finite time.
\end{lem}
\begin{proof}
It is clear that in the case where the configuration is of type {\tt Undefined2}, we are sure to reach a configuration of type {\tt Leader} in a finite time, since there is only one robot that is the closest to one corner (this robot will move until it reaches the closest corner). It is also clear that in the case where the configuration is of type {\tt Undefined1}, either a configuration of type {\tt Undefined2} is reached and hence a configuration of type {\tt Leader} is eventually reached or a configuration where there are two robots that are both the closest to a corner is reached, this case is part of the cases below:

\begin{enumerate}
\item There are exactly two robots that are the closest to one corner (let these two robots be $\mathcal R1$ and $\mathcal R2$ respectively). In this case, $\mathcal R3$ will be used to break the symmetry: In the case $Dist(\mathcal R1,\mathcal R3)=Dist(\mathcal R2,\mathcal R3)$, $\mathcal R3$ will be the one allowed to move, it destination is its adjacent node towards either ($a$) $\mathcal R1$ or $\mathcal R2$ if $Dist(\mathcal R1,\mathcal R3)>1$. Or ($b$) its adjacent free node from which its distance to $\mathcal R1$ will be different from its distance to $\mathcal R2$. In both cases ($a$ and $b$), we reach a configuration where $Dist(\mathcal R1,\mathcal R3) \ne Dist(\mathcal R2,\mathcal R3)$. In the case $Dist(\mathcal R1,\mathcal R3) \ne Dist(\mathcal R2,\mathcal R3)$, the robot that is the closest to $\mathcal R3$ will be the one allowed to move, its destination is its adjacent free node on a shortest path towards the corner. Note that once it has moved, either it reaches the corner or it becomes the closest one. Thus we are sure that a configuration of type {\tt Leader} is reached in a finite time. 
\item All the robots are the closest to a corner. If the
  configuration is of type {\tt Undefined4-1}, then there will be one
  robot that will be allowed to move (the one that is on a
  borderline), once it has moved, it becomes the closest to one corner
  of the grid, thus we are sure to reach a configuration of type {\tt
    Leader} in a finite time. In the case there are two robots at a
  borderline, The third robot (which is not on a borderline) is the
  one that will move becoming the closest robot to one corner of the
  grid. Thus in this case too, we are sure to reach a configuration of
  type {\tt Leader}. In the case all the robots are on a borderline
  then, i) if there is more than one robot on the same borderline
  (note that in this case the borderline contains two robots), the
  robot that is not part of the borderline moves towards the
  closest corner becoming the closest one, thus we are sure that a
  configuration of type {\tt Leader} is reached in a finite time. In
  the case there is one robot at each borderline, then one robot is
  easily elected to move becoming the closest to one corner of the
  grid. Thus, in this case too we are sure to reach a configuration of
  type {\tt Leader} in a finite time. In the case there is no robot on
  the borderline. If there are two robots that are the closest to the
  same corner such as the third robot is the only closest robot to
  another corner then this robot is the one allowed to move, when it
  does it becomes the only one that is the closest to one corner of
  the grid. Thus we are sure to reach a configuration of type {\tt
    Leader}. In the case there is one robot ($\mathcal R3$) that is the closest
  to both corners $C1$ and $C2$ such as $\mathcal R1$ and $\mathcal R2$ are also the
  closest to $C1$ $C2$ respectively, then $\mathcal R3$ is the one allowed to
  move towards one of the closest corner. Note that once it has moved,
  it becomes the closest one and hence we are sure that a
  configuration of type {\tt Leader} is reached in a finite time. In
  the case all the robots are the closest to different corner, we are
  sure that one of them is the closest one to one corner that is
  between the two other target corners (the closest to the other
  robots). This robot is the one allowed to move, its destination is
  its adjacent free node towards the closest corner. Note that one it
  moves it becomes either even closer (and hence it will be the only
  one that can move) or it will reach the corner. In both cases we are
  sure that a configuration of type {\tt Leader} is reached.
\end{enumerate} 

From the cases above we can deduce that starting from a configuration of type {\tt Undefined}, a configuration of type {\tt Leader} is reached in a finite time and the lemma holds.

\end{proof}

\begin{lem}\label{lem:setup}
Starting from any towerless configuration, a configuration of type {\tt Set-Up} is reached in a finite time. 
\end{lem}
\begin{proof}
From Lemma \ref{Leader}, \ref{Choice} and \ref{Undefined} we can deduce that starting from any arbitrary towerless configuration that does not contain a line of robots on the longest line of the grid, a configuration of type {\tt Set-Up} is reached in a finite time and the lemma holds. 
\end{proof}

\paragraph{Orientation Phase.}

In this phase, an orientation of the grid is determined in the
following manner: The starting configuration contains a line of robots
on one of the longest borderline (of length greater than 3) starting
at one of its corner. The robot which is at the corner is the one
allowed to move, its destination is its adjacent occupied node. Once
it has moved, a tower is created. Then, we can determine a
coordination system where each node has unique coordinates, see Figure
\ref{EXPP}, page \pageref{EXPP}. The node with coordinates $(0,0)$ is
the unique corner that is the closest to the tower. The X-axis is
given by the vector linking the node $(0,0)$ to the node
where the tower is located. The Y-axis is given by the vector linking
the node $(0,0)$ to its neighboring node that does not contain the
tower.

The following lemma is straightforward:

\begin{lem}\label{lem:oriented}
  Starting from a configuration of type {\tt Set-Up}, a configuration
  of type {\tt Oriented} is reached in one step.
\end{lem}

\paragraph{Exploration Phase.}
 \begin{wrapfigure}{r}{0.3\textwidth}
   \centering
   \includegraphics[width=3.5cm]{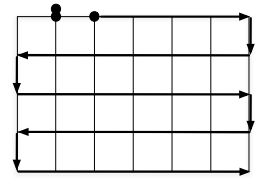}
   \caption{Exploration phase}\label{Exploration}
 \end{wrapfigure} 

 This phase starts from an {\tt Oriented} configuration.  Note that,
 once this configuration is reached, nodes of coordinates $(0,0)$,
 $(0,1)$, and $(0,2)$ have been necessarily visited. Then, the goal
 is to visit all other nodes.  To ensure that the exploration phase
 remains distinct from the previous phases and keep the coordinate
 system, we only authorize the robot that is single on a node to
 move. This robot is called the {\em explorer}.

 To explore all remaining nodes, the explorer should order all
 coordinates in such a way that (a) $(0,0)$ and $(0,1)$ are before its
 initial position (that is $(0,2)$) and all other coordinates are
 after; and (b) for all non-maximum coordinates $(x,y)$, if $(x',y')$
 is successor of $(x,y)$ in the order, then the nodes of
 coordinates $(x,y)$ and $(x',y')$ are neighbors. Such an order can be
 defined as follows:
$$(a,b) \preceq (c,d) \equiv b<d \vee [b=d \wedge ((a=c) \vee (b\bmod 2 = 0 \wedge a<c) \vee (b\bmod 2 = 1 \wedge a>c)]$$

Using the order $\preceq$, the explorer moves as follows: While the
explorer is not located at the node having the maximum coordinates
according to $\preceq$, the explorer moves to the neighbor whose
coordinates are successors of the coordinates of its current
position, as described in Figure~\ref{Exploration}.

The following lemma is straightforward:

\begin{lem}\label{lem:explore}
  The {\tt Exploration} phase terminates in finite time and once terminated all nodes have been visited.
\end{lem}

By Lemmas \ref{lem:notower}-\ref{lem:explore}, follows:

\begin{theo}The deterministic exploration of any $(i,j)$-Grid with
  $j>3$ can be solved in CORDA using 3 oblivious robots and the three
  phases {\tt Set-Up}, {\tt Orientation}, and {\tt Exploration}.
\end{theo}

\subsection{Exploring a (2,3)-Grid}\label{2-3}

The idea of the solution for the $(2,3)$-Grid is rather simple.
Consider the two longest borderlines of the grid.  Since there are
initially three isolated robots on the grid, there exists one of the
two longest borderlines, say $D$, that contains either all the robots
or exactly two robots. In the second case, the robot that is not part
of $D$ moves to the adjacent free node on the shortest path towards
the free node of $D$. Thus, the three robots are eventually located on
$D$.  Next, the robot not located on any corner moves to one of its
two neighboring occupied nodes (the destination is chosen by the
adversary).  Thus, a tower is created.  Once the tower is created, the
grid is oriented.  Then, the single robot moves to the adjacent free
node in the longest borderline that does not contain any tower.  Next,
it explores the nodes of this line by moving in towards the tower.
When it becomes neighbor of the tower, all the nodes of the
$(2,3)$-Grid have been explored.

The following theorem is straightforward.

\begin{theo}
  The deterministic exploration of a $(2,3)$-Grid can be solved in
  CORDA using 3 oblivious robots.
\end{theo}

\section{Deterministic solution for a (3,3)-grid using five robots}\label{sec:5R}

\begin{figure}[htb]
  \begin{center}
  \epsfig{figure=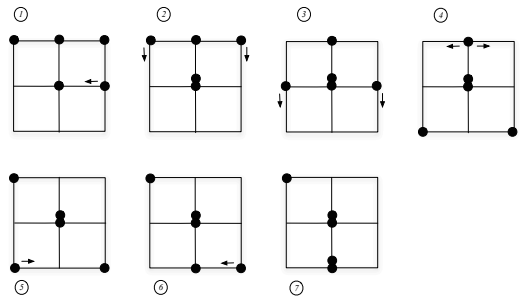,width=7cm}\\
  \caption{Exploration task on grids $(3,3)$}\label{EXP5-3C1}
  \end{center}
\end{figure}

In this section, we propose an algorithm that explores using five
robots the $(3,3)$-Grid, in CORDA and assuming weak multiplicity
detection.  The algorithm works in two phases, the {\tt Exploration}
phase and the {\tt Preparation} phase.  Figures~\ref{EXP5-3C1}
and~\ref{EXP5-3C2} depict the {\tt Exploration} phase.

The {\tt Exploration} phase starts from any of the three special
configurations shown in Figure~\ref{EXP5-3C1}-Case~$(1)$,
Figure~\ref{EXP5-3C2}-Case$(1a$), and
Figure~\ref{EXP5-3C2}-Case$(1b$), respectively.  In the former case,
the unique robot that is (1) on a borderline, (2) not at a corner, and
(3) not on the borderline linking the two occupied corners, moves
toward the center.  In Case~$(1a)$ of Figure~\ref{EXP5-3C2}, the
unique robot located at a corner moves toward one of its neighbors
(chosen by the adversary).  Similarly, in Case~$(1b)$ in
Figure~\ref{EXP5-3C2}, the robot located at the center moves toward
one of its neighbors.  In the three cases, one tower is created and
the system reaches Case~$2$ of either Figure~\ref{EXP5-3C1} or
Figure~\ref{EXP5-3C2}, depending on the initial configuration.  Next,
the exploration is made following the moves depicted in either
Figure~\ref{EXP5-3C1} or Figure~\ref{EXP5-3C2}, respectively.

\begin{figure}[H]
 \begin{center}
 \epsfig{figure=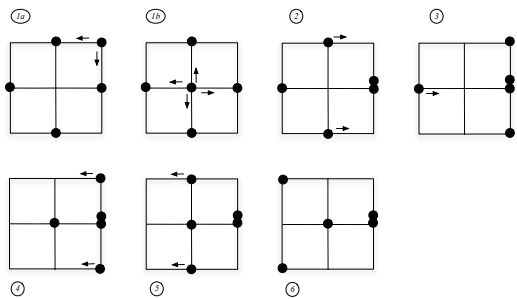,width=7cm}\\
  \caption{Special Exploration of grids $(3,3)$}\label{EXP5-3C2}
  \end{center}
\end{figure}

The {\tt Preparation} phase starts from any towerless configuration
that is not one of the three initial configurations of the exploration
phase. The {\tt Preparation} phase aims at reaching one of these three
configurations.  The detailed algorithm of this phase is left as an
exercise for the reader --- a solution is given in the appendix.

\begin{theo}\label{theo5}
  The deterministic exploration of a $(3,3)$-Grid can be solved in 
  CORDA using 5 oblivious robots.
\end{theo}

\section{Conclusion}
\label{conclu}

We presented necessary and sufficient conditions to explore a
grid-shaped network with a team of $k$ asynchronous oblivious
robots. Our results show that, perhaps surprisingly, exploring a grid
is easier than exploring a ring. In the ring, deterministic
(respectively, probabilistic) solutions essentially require five
(resp., four) robots. In the grid, three robots are necessary (even in
the probabilistic case) and sufficient (even in the deterministic
case) in the general case, while particular instances of the grid do
require four or five robots.  Note that the general algorithm given in
that paper requires exactly three robots.  It is worth investigating
whether exploration of a grid of $n$ nodes can be achieved using any
number $k$ ($3> k \geq n-1$) of robots, in particular when $k$ is
even.




\bibliographystyle{plain} 
\bibliography{biblio}

\pagebreak

\appendix

\setcounter{page}{1}
\pagenumbering{roman}
\pagestyle{plain}
\theoremnumbering{roman}

\def\thelemma{\Roman{lemma}}
\setcounter{lemma}{0}
\def\thedefinition{\Roman{definition}}
\setcounter{definition}{0}
\def\thecorollary{\Roman{corollary}}
\setcounter{corollary}{0}
\def\thetheorem{\Roman{theorem}}
\setcounter{theorem}{0}
\def\theremark{\Roman{remark}}
\setcounter{remark}{0}
\def\theproperty{\Roman{property}}
\setcounter{property}{0}
\def\thehypothesis{\Roman{hypothesis}}
\setcounter{hypothesis}{0}
\def\theproposition{\Roman{proposition}}
\setcounter{proposition}{0}
\def\thespecification{\Roman{specification}}
\setcounter{specification}{0}

\section{Preparation phase of the algorithm working with $5$ robots in the $(3,3)$-Grid}

The aim of the {\tt Preparation} phase is to reach one of the special
configurations, where the {\tt Exploration} phase can start.  It
starts from an arbitrary towerless configuration that is not one of the
three initial configurations shown in either Figure~\ref{EXP5-3C1} or
Figure~\ref{EXP5-3C2}.

 \begin{figure}[hbp] 
 \begin{minipage}[b]{.46\linewidth}
  \begin{center}
  \epsfig{figure=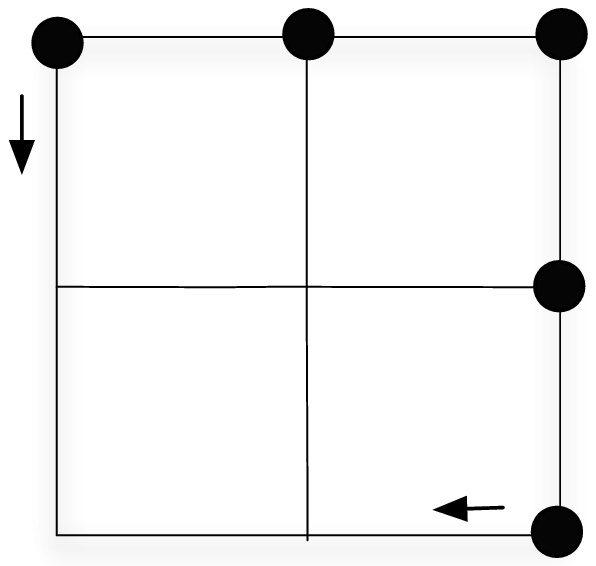,width=2cm}\\
  \caption{Configuration $(3,1,1)$}\label{PerfectSym}
  \end{center}
 \end{minipage} \hfill
 \begin{minipage}[b]{.46\linewidth}
  \begin{center}
  \epsfig{figure=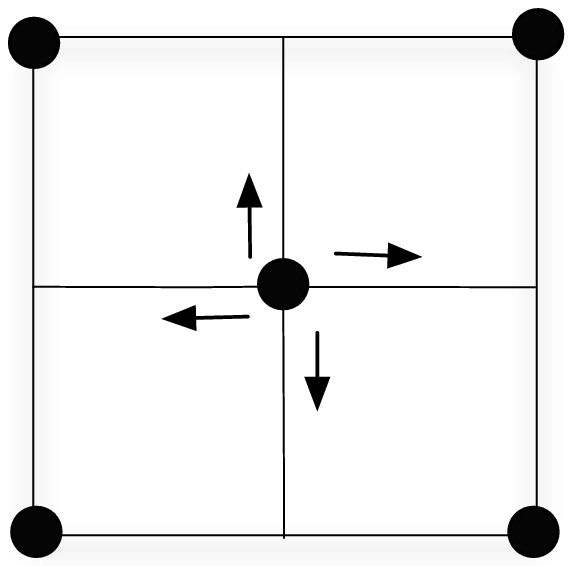,width=2cm}\\
  \caption{Instance of a configuration $(2,1,2)$}\label{Sym212}
  \end{center}
 \end{minipage} \hfill
 \begin{minipage}[b]{.46\linewidth}
  \begin{center}
  \epsfig{figure=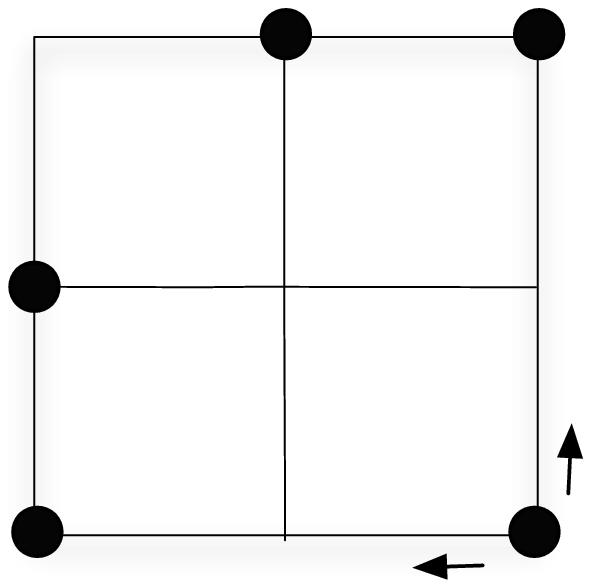,width=2cm}\\
  \caption{Instance of a configuration $(2,1,2)$}\label{2-1-2(1)}
  \end{center}
 \end{minipage} \hfill
 \begin{minipage}[b]{.46\linewidth}
  \begin{center}
  \epsfig{figure=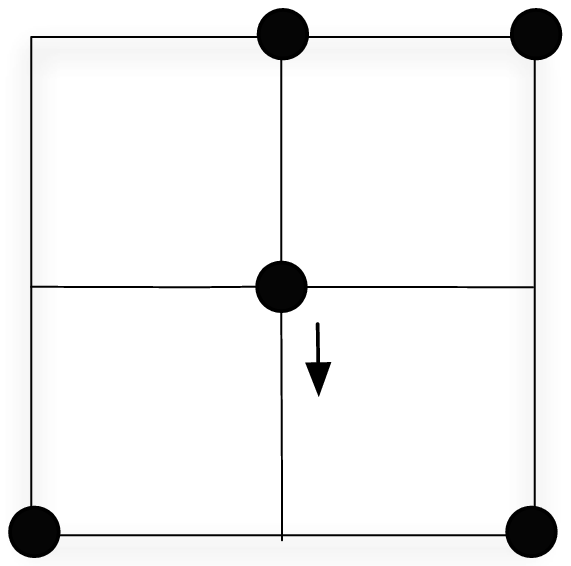,width=2cm}\\
  \caption{Instance of a configuration $(2,1,2)$}\label{2-1-2(2)}
  \end{center}
 \end{minipage} \hfill
\end{figure}

Let us define some terms that will be used later: let the
interdistance $d$ be the minimal distance among distances between each
pair of robots. We call a d.block a sequence of consecutive robots
that are at distance $d$. The size of an 1.block is the number of
robots it contains. We refer to a configuration by a set of three
values $(X1,X2,X3)$ such as $Xi$ represents the number of robots on
the line $i$. Note that X1 and X3 are borderlines. Since the grid is
of size $(3,3)$, we do not know which borderlines correspond to $X1$
and $X3$.  Some ambiguities can appear and thus for the same
configuration there will be many possible sequences $(X1,X2,X3)$. The
robots could be confused not knowing which action to take. To avoid
this situation, we will use the following method: First we will choose
one or two guide lines in the following manner: the line that contains
the d.biggest d.block of robots is elected as a guide line. Note that
the guide line can only contain two or three robots. In the case there
are two possible guide lines that are perpendicular to each other,
then i) in the case only one of this two guide lines is at the
borderline of the grid, then this line is the guide line.  ii) In the
other case, the guide line is elected as follow: Let $D1$ be one
possible guide line and $D2$, $D3$ be the lines that are horizontal to
$D1$. In the same manner let $D'1$ be the other possible guide line
and $D'2$, $D'3$ be the lines that are horizontal to $D'1$. Let $B$ be
the number of the biggest d.blocks on the lines $Di$ and $B'$ be the
number of the biggest d.blocks on the lines $D'i$. The guide line is
the one corresponding to the biggest value among $B$ and $B'$. For
Instance in Figure~\ref{Explanation}, the configuration can be
$(2,1,2)$ or $(2,2,1)$. We can see that $d=1$, and the size of the
biggest 1.block is equal to $2$. Note that there is an 1.block of size
$2$ on two borderlines that are perpendicular to each other (on $D3$
and $D'1$ ---refer to Figure~\ref{Explanation}). Let $B$ be the number
of 1.blocks on the lines that are horizontal to $D3$, clearly
$B=2$. In the same manner, let $B'$ be the number of 1.blocks of size
$2$ on the lines that are horizontal to $D'1$ (clearly $B'=1$). We can
see that $B>B'$, thus the guide lines are both $D3$ and $D1$ (The
lines that are considered are the ones that are horizontal to $D3$ and
$D1$). Thus the configuration is of type $(2,1,2)$.

\begin{figure}[hbp] 
 \begin{center}
 \epsfig{figure=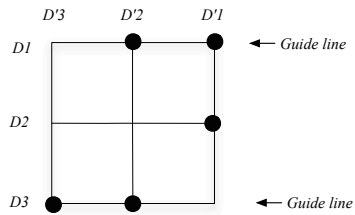,width=5cm}\\
  \caption{Guide-lines, configuration of type $(2,1,2)$}\label{Explanation}
  \end{center}
\end{figure}

\begin{figure}[hbp]
  \begin{center}
  \epsfig{figure=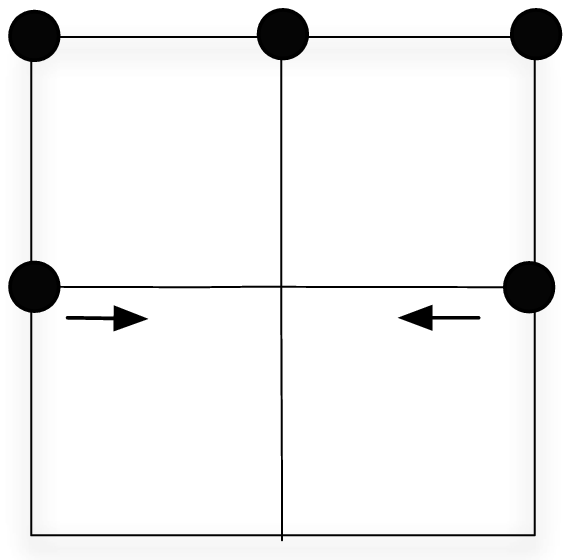,width=3cm}\\
  \caption{Instance of a configuration $(3,2,0)$}\label{3-2(1)}
  \end{center}
  \end{figure}

The triple set $(X1,X2,X3)$ refer then to the number of robots that are horizontal to the guide lines. The following cases are then possible:

\begin{itemize}
\item The configuration is of type $(1,1,3)$. Two sub-cases are possible: 
i) The configuration is similar to the one shown in Figure \ref{PerfectSym}. It is clear that in this case no guide line can be determined. The robots allowed to move are the ones that are at the corner having one free node as a neighbor, their destination is their adjacent free node on the borderline they belong to. ii) The remaining cases: One line can be elected as the guide line, this line is the one that contains an 1.block of size $3$ ($X3$). The robot that is alone on the borderline ($X1$) is the one allowed to move, its destination is its adjacent free node on the shortest path towards the middle line (the one that contains $X2$). Note that in a case of symmetry, the adversary will break the symmetry by choosing one of the two possible neighboring nodes.

\item The configuration is of type $(1,2,2)$. 
 The robot that is alone on the borderline ($X1$) is the one allowed to move, its destination is its adjacent free node on the shortest path towards the free node on the line that contains $X2$. 

\item The configuration is of type $(1,3,1)$. Two sub-cases are possible: i) The configuration is similar to the one shown in Figure \ref{EXP5-3C2}, Step 1. Note that for this configuration, there is a dedicated algorithm that solves the exploration problem. The algorithm is detailed in Figure \ref{EXP5-3C2}. Note that since the system is asynchronous, the adversary in some steps of the algorithm can activates one of the two robots that are allowed to move. In this case, the robot that was supposed to move in the first place is the only one that can move, thus by moving the configuration reached when both robots were activated is reached again ii) The remaining cases: we are sure that there is one robot that is part of an 1.block of size $3$ (in the middle line) that has two neighboring free nodes (Note that there is only five robots and a single 1.block of size $3$), let this robot be $\mathcal R1$. $\mathcal R1$ is the only one allowed to move, its destination is its adjacent free node towards the closest robot that is in one of the two borderlines that are horizontal to the 1.block of size $3$.
\item The configuration is of type ($2,1,2)$.  Note that the configuration does not contain an 1.block of size $3$. Let $D1$ and $D3$ be the two borderlines corresponding to $X1$, $X2$ respectively. The sub-cases below are possible:
\begin{itemize}
\item Both $D1$ and $D2$ contains robots at distance $2$ ($d=2$). In this case, we are sure that there is one robot on the center of the grid (on the middle of the middle line, otherwise the configuration will contains an 1.block of size $3$). This robot is the one allowed to move, its destination is one of adjacent free node towards the borderline (refer to Figure \ref{Sym212}).
\item The robots on $D1$ are at distance $1$ and the robots on $D2$ are at distance $2$. If the robot that is in the middle line (according to the guide line) is also on a borderline (see Figure \ref{2-1-2(1)}), we are sure that there is one robot at the corner of the grid not having any neighboring robot. This robot is the one allowed to move, its destination is one of its adjacent free node. If the robot is in the center of the grid (see Figure \ref{2-1-2(2)}), then this robot is the one allowed to move its destination is its adjacent free node towards $D2$.
\item Both $D1$ and $D2$ contains robots at distance $1$ ($d=1$). Let $D3$ be the middle line that is horizontal to both $D1$ and $D2$. The robot allowed to move is the one that is on $D3$, its destination is its adjacent node towards $D1$ or $D2$ (The scheduler will make the choice in the case of symmetry). 
\end{itemize} 
\item The configuration is of type $(2,3,0)$. In this case the robot that in the middle line that contains three robots having an free node as a neighbor on the line that contains two robots is the one allowed to move, its destination is this adjacent free node.
\item The configuration is of type $(3,0,2)$. In this case the robots that are in $X3$ (the line that contains two robots) are the one allowed to move, their destination is their adjacent free node on the shortest path towards $X2$.
\item The configuration is of type $(3,2,0)$ but is different from the special configuration (refer to Figure \ref{3-2(1)}). The robots allowed to move are the two robots that are on the line corresponding to $X2$. Their destination is their adjacent free node on the line that contains $X2$. Its is clear that in the case the adversary activates only one of these two robots the configuration reached will be the Special configuration (see Figure \ref{EXP5-3C1}, step 1), Thus the exploration task can be performed as shown in \ref{EXP5-3C1}. In the case the adversary activates both robots at the same time, then a tower is created and the configuration reached is like the one shown in Figure \ref{EXP5-3C1}, step 2. In this case too the exploration can be performed. 
\end{itemize} 

Note that once one of the two special configurations is built, one tower is created and the exploration task can be performed. refer to Figures \ref{EXP5-3C1} and \ref{EXP5-3C2}.

\paragraph{Correctness Proof.}

\begin{lem}\label{1-2-2}
Starting from a configuration of type $(1,2,2)$, a configuration of type $(2,3,0)$ is reached in a finite time.
\end{lem} 

\begin{proof}
In a configuration of type $(1,2,2)$ the robot that is allowed to move is the one that is alone on the borderline containing $X1$, let $\mathcal R1$ be this robot, its destination is its adjacent free node towards $X2$, Since line $X2$ contains two robots, when $\mathcal R1$ joins $X2$, $X2$ will contain an 1.block of size $3$ and $X1$ will contain no robot. Thus the configuration reached is of type $(2,3,0)$ and the lemma holds.  
\end{proof}

\begin{lem}\label{1-3-1}
Starting from a configuration of type $(1,3,1)$, either a configuration of type $(2,2,1)$ or of type $(2,1,2)$ is reached in a finite time. 
\end{lem}

\begin{proof}
When the configuration is of type $(1,3,1)$, we are sure that there is one robot that is part of the 1.block of size $3$ on $X2$ that has two neighboring free nodes. This robot is the one allowed to move its destination is its adjacent free node towards the closest robot on either $X1$ or $X2$. Suppose that such a robot is the one that is in the middle of the 1.block of size $3$. Once the robot has  moved, the configuration becomes of type $(2,1,2)$ and the lemma holds. If such a robot is at the extremity of the 1.block of size $3$, then by moving, the configuration reached is of type $(2,2,1)$ and the lemma holds.
\end{proof}

\begin{lem}\label{2-1-2}
Starting from a configuration of type $(2,1,2)$, a configuration of type $(3,0,2)$ is reached in a finite time.
\end{lem}

\begin{proof}
The cases below are possible:
\begin{enumerate}
\item Both $D1$ and $D2$ contains robots at distance $2$ ($d=2$). It is clear that in this case there is one robot that is in the center of the grid. This robot is the one allowed to move, its destination is one of its adjacent free node. By moving, the robot join a borderline. Note that this borderline contains an 1.block of size $3$. Thus the configuration reached will be $(3,0,2)$.
\item The robots on $D1$ are at distance $1$ and the robots on $D2$ are at distance $2$. In this case the robot that is on the borderline on $D2$, being at the corner of the grid and not having any neighboring robot is the one that moves towards one of its adjacent free node. Note that once the robot has moved, the configuration reached remains of type $(2,1,2)$, however, both $D1$ and $D2$ contains robots at distance $1$.

\item Both $D1$ and $D2$ contains robots at distance $1$ ($d=1$). Let $D3$ be the middle line that is horizontal to both $D1$ and $D2$. In this case the robot that is on $D3$ is the one allowed to move, its destination is its adjacent free node towards one of the two neighboring borderlines that contain an 1.block of size $2$. Note that we are sure that this robot has at least one free node as a neighbor otherwise the configuration contains a single 1.block of size $3$ and the configuration will not be of type $(2,1,2)$. Once the robot has moved, a new 1.block of size $3$ is created at one borderline and the configuration will be of type $(3,0,2)$.  
\end{enumerate}

From the cases above, we can deduce that starting from a configuration of type $(2,1,2)$, a configuration of type $(3,0,2)$ is reached in a finite time and the lemma holds.
\end{proof}

\begin{lem}\label{2-3-0}
Starting from a configuration of type $(2,3,0)$, a configuration of type $(3,2,0)$ is reached in a finite time.
\end{lem}

\begin{proof}
When the configuration is of type $(2,3,0)$, the robot allowed to move is the one that is on the line that contains $X2$ having an free node as a neighbor on the line that contains two robots. Note that once the robot has moved, a new 1.block of size $3$ is created one borderline of the grid. Thus the configuration reached will be of type $(3,2,0)$ and the lemma holds. 
\end{proof}

\begin{lem}\label{3-0-2}
Starting from a configuration of type $(3,0,2)$, either a configuration of type $(3,2,0)$ or of type $(3,1,1)$ is reached in a finite time.
\end{lem}

\begin{proof}
When the configuration is of type $(3,0,2)$, the robots that are on the line that $X3$ are the one allowed to move. When they do, they move to their adjacent free node towards the line that is horizontal to the the one containing an 1.block of size $3$. Note that in the case the adversary activates both robots allowed to move at the same time, then the configuration reached is of type $(3,2,0)$ and the lemma holds. If it is not the case, the configuration reached is of type $(3,1,1)$ and the lemma holds.
\end{proof}

\begin{lem}\label{3-1-1}
Starting from a configuration of type $(3,1,1)$, either a configuration of type $(3,2,0)$ or of type $(2,2,1)$ is reached in a finite time.
\end{lem}

\begin{proof}
In the case the configuration is similar to the one shown in Figure \ref{PerfectSym}. The robots that are at the corner having an free node as a neighbor are the one allowed to move. Their destination is their adjacent free node. Note that in the case the adversary activates both robots at the same time, the configuration reached is of type $(2,2,1)$ and the lemma holds. In the case the adversary activates only one robot, then the configuration reached remains of type $(3,1,1)$ but it is different from the Figure \ref{PerfectSym}. For the other configurations of type $(3,1,1)$ (all the configurations that are different from the one shown in Figure \ref{PerfectSym}). The robot that is allowed to move is the one that is single on the borderline that contains $X3$. Its destination is its adjacent free node on the shortest path towards the line that contains $X2$. Note that once it has moved, the configuration reached is of type $(3,2,0)$ and the lemma holds.
\end{proof}

\begin{lem}\label{1-2-2bis}
Starting from a configuration of type $(1,2,2)$, a configuration of type $(3,2,0)$ is reached in a finite time.
\end{lem}

\begin{proof}
From Lemma \ref{1-2-2}, we are sure that starting from a configuration of type $(1,2,2)$, a configuration of type $(2,3,0)$ is reached in a finite time. From Lemma \ref{2-3-0} we are sure that starting from a configuration of type $(2,3,0)$, a configuration of type $(3,2,0)$ is reached in a finite time. Thus we can deduce that starting from a configuration of type $(1,2,2)$, a configuration of type $(3,2,0)$ is reached in a finite time and the lemma holds.
\end{proof}

\begin{lem}\label{1-3-1bis}
Starting from a configuration of type $(1,3,1)$, a configuration of type $(3,2,0)$ is reached in a finite time.
\end{lem}

\begin{proof}
From Lemma \ref{1-3-1}, we are sure that starting from a configuration of type $(1,3,1)$, a configuration of type $(2,2,1)$ is reached in a finite time. From Lemma \ref{1-2-2bis} we are sure that starting from a configuration of type $(1,2,2)$, a configuration of type $(3,2,0)$ is reached in a finite time. Thus we can deduce that starting from a configuration of type $(1,3,1)$, a configuration of type $(3,2,0)$ is reached in a finite time and the lemma holds.
\end{proof}

\begin{lem}\label{2-1-2bis}
Starting from a configuration of type $(2,1,2)$, a configuration of type $(3,2,0)$ is reached in a finite time.
\end{lem}

\begin{proof}
From Lemma \ref{2-1-2}, we are sure that starting from a configuration of type $(2,1,2)$, a configuration of type $(3,0,2)$ is reached in a finite time. From Lemma \ref{3-0-2}, we are sure that starting from a configuration of type $(3,0,2)$, a configuration of type $(3,2,0)$ is reached in a finite time. Thus we can deduce that starting from a configuration of type $(2,1,2)$, a configuration of type $(3,2,0)$ is reached in a finite time and the lemma holds.
\end{proof}

\begin{lem}
Starting from any configuration that is towerless, a configuration of type $(3,2,0)$ is reached in a finite time.
\end{lem}

\begin{proof}
From Lemmas \ref{2-3-0}-\ref{2-1-2bis}, we can deduce that starting from any configuration that is towerless, a configuration of type $(3,2,0)$ is reached in a finite time and the lemma holds.
\end{proof}

\begin{lem}\label{final}
Starting from one of the three special configurations, all the nodes of the grid are explored and the algorithm stops.
\end{lem}

\begin{proof}
 It is easy to see from Figure \ref{EXP5-3C1}, that all the nodes of the grid are explored. Thus the lemma holds.
\end{proof}

\begin{lem}
Starting from any configuration of type $(3,2,0)$, the exploration can be performed.
\end{lem}

\begin{proof}
If the configuration is the special configuration (refer to Figure \ref{EXP5-3C1} (step 1)), then according to Lemma \ref{final}, the exploration task is performed and all the nodes of the ring are explored. If the configuration is as the one show in Figure \ref{3-2(1)}, then the two robots that are not part of the 1.block of size $3$ are the one allowed to move, their destination is their adjacent node in the center of the grid. In the case where the adversary activates only one of the two robots allowed to move, the special configuration is reached and the lemma holds. If both robots are activated then a tower is created in the center of the grid and the configuration reached will be as the one shown in Figure \ref{EXP5-3C1} (Step2) and in this case too the exploration is performed and the lemma holds.
\end{proof}

From the lemmas above we can deduce that:

\bigskip

\noindent {\bf Theorem \ref{theo5}}\ {\em The deterministic
  exploration of a $(3,3)$-Grid can be solved in CORDA using 5
  oblivious robots.}

\end{document}